\documentclass[format=acmsmall, review=false, authorversion=true, screen=true]{acmart}

\pdfoutput=1

\usepackage[T1]{fontenc}
\usepackage[utf8]{inputenc}
\usepackage[linesnumbered,ruled,vlined]{algorithm2e}
\usepackage{listings}
\usepackage{hyperref}
\usepackage{graphicx}
\usepackage{subfigure}
\usepackage{multirow}
\usepackage{tikz}

\usepackage{macros-custom}

\newcommand{\hgTotal}{3,648}
\newcommand{\hgNew}{3,142}
\newcommand{\cqOld}{424}
\newcommand{\cspOld}{82}

\newcommand{\cqApp}{1113}
\newcommand{\cqRand}{500}
\newcommand{\cspTot}{2035}
\newcommand{\cspApp}{1090}
\newcommand{\cspOth}{82}
\newcommand{\cspRand}{863}

\newcommand{\cqAppKOneYes}{673}
\newcommand{\cqAppKOneNo}{440}
\newcommand{\cqAppKTwoYes}{432}

\newtheorem{myTheorem}{Theorem}
\newtheorem{myLemma}{Lemma}
\newtheorem{myDefinition}{Definition}

\newenvironment{myintro}%
{\list{}{\leftmargin=0.1in\rightmargin=0.1in}\item[]}%
{\endlist}

\AtBeginDocument{%
  \providecommand\BibTeX{{%
    \normalfont B\kern-0.5em{\scshape i\kern-0.25em b}\kern-0.8em\TeX}}}

\begin{document}
\setcopyright{none}
\settopmatter{printacmref=false,printccs=false,printfolios=true}	

\title{HyperBench: A Benchmark and Tool for Hypergraphs and Empirical Findings}
\titlenote{This is an extended and enhanced version of the paper ``HyperBench: 
A Benchmark and Tool for Hypergraphs and Empirical Findings" presented at PODS 2019~\cite{DBLP:conf/pods/FischlGLP19}}

\author{Wolfgang Fischl}
\affiliation{
	\institution{TU Wien}
}
\email{wfischl@dbai.tuwien.ac.at}

\author{Georg Gottlob}
\affiliation{
	\institution{University of Oxford}
}
\email{georg.gottlob@cs.ox.ac.uk}

\author{Davide Mario Longo}
\affiliation{
	\institution{TU Wien}
}
\email{dlongo@dbai.tuwien.ac.at}

\author{Reinhard Pichler}
\affiliation{
	\institution{TU Wien}
}
\email{pichler@dbai.tuwien.ac.at}


\begin{abstract}
To cope with the intractability of answering Conjunctive Queries (CQs) and solving Constraint Satisfaction Problems (CSPs), several notions of hypergraph decompositions have been proposed -- giving rise to different notions of width, noticeably, plain, generalized, and fractional hypertree width (hw, ghw, and fhw). Given the increasing interest in using such decomposition methods in practice, a publicly accessible repository of decomposition software, as well as a large set of benchmarks, and a web-accessible workbench for inserting, analyzing, and retrieving hypergraphs are called for. 

We address this need by providing (i) concrete implementations of hypergraph decompositions (including new practical algorithms), (ii) a new, comprehensive benchmark of hypergraphs stemming from disparate CQ and CSP collections, and (iii) HyperBench, our new web-inter\-face for accessing the benchmark and the results of our analyses. In addition, we describe a number of actual experiments we carried out with this new infrastructure.
\end{abstract}


\maketitle

\lstset{language=SQL,numbers=left}

\section{Introduction}
\label{sec:introduction}

In this work we study computational problems on hypergraph decompositions 
which are designed to speed up 
the evaluation of Conjunctive Queries (CQs) 
and the solution of Constraint Satisfaction Problems (CSPs).
Hypergraph decompositions have 
meanwhile found their way into commercial database systems such as LogicBlox 
\cite{DBLP:conf/sigmod/ArefCGKOPVW15,
	DBLP:journals/tods/OlteanuZ15,DBLP:journals/pvldb/BakibayevKOZ13,DBLP:journals/tods/KhamisNRR16,DBLP:conf/pods/KhamisNR16} and advanced research prototypes 
such as 
EmptyHeaded~\cite{DBLP:journals/tods/AbergerLTNOR17,DBLP:journals/corr/AbergerTOR16,DBLP:conf/sigmod/TuR15,DBLP:conf/sigmod/PerelmanR15}.  
Hypergraph decompositions have also been successfully used in 
the CSP area \cite{DBLP:journals/aicom/AmrounHA16,DBLP:journals/jetai/HabbasAS15,DBLP:conf/aiia/LalouHA09}.
In theory, 
the pros and cons of various notions of decompositions and widths are well understood 
(see \cite{DBLP:conf/pods/GottlobGLS16} for a survey).
However, from a practical point of view, many questions have remained~open. 

We want to collect hypergraphs  from different application contexts, 
analyze their structural properties and, in particular, their (generalized) hypertree width ($\ghw$ and $\hw$, respectively) 
and make this hypergraph collection together with the results of our analyses publicly available. 
The investigation of millions of CQs \cite{DBLP:journals/pvldb/BonifatiMT17,DBLP:conf/sigmod/PicalausaV11}
posed at various SPARQL endpoints suggests that these real-world CQs with atoms of arity $\leq 3$ 
have very low $\hw$: the overwhelming majority is acyclic; almost all of the rest
has $\hw = 2$. It is, however, not clear if CQs with arbitrary arity and CSPs also have 
low hypertree width, say, $\hw \leq 5$.
Ghionna et al.~\cite{DBLP:conf/icde/GhionnaGGS07} gave a positive answer to this question 
for a small set of TPC-H benchmark queries. We significantly extend their collection of CQs. 

Answering CQs and solving CSPs are fundamental tasks in Computer Science. Formally, they are the same problem, since both  correspond to the evaluation of first-order formulae over a finite structure, such that the formulae only use $\{\exists,\wedge\}$ 
as connectives
but not $\{\forall, \vee, \neg\}$.
Both problems, answering CQs and solving CSPs, 
are NP-complete \cite{DBLP:conf/stoc/ChandraM77}. 
Consequently, the search for tractable fragments of these problems has been an active research area in the database and artificial intelligence communities for several decades. 

The most powerful methods 
known to date 
for defining tractable fragments are based on various decompositions of the 
hypergraph structure 
underlying 
a given CQ or CSP. 
The most important forms of decompositions 
are 
{\it hypertree decompositions (HDs)\/}  
\cite{DBLP:journals/jcss/GottlobLS02}, 
{\it generalized 
	hypertree decompositions (GHDs)}~\cite{DBLP:journals/jcss/GottlobLS02}, and 
{\it fractional hypertree decompositions (FHDs)}~\cite{DBLP:journals/talg/GroheM14}.
These decomposition methods give rise to three
notions of width of a hypergraph $H$:
the  {\it hypertree width} $\hw(H)$,  
{\em generalized hypertree width} $\ghw(H)$, 
and {\em fractional hypertree width} 
$\fhw(H)$, 
where, $\fhw(H)\leq \ghw(H)\leq \hw(H)$ holds for every 
hypergraph $H$. 
For definitions, see Section~\ref{sec:preliminaries}. 

Both, answering CQs and solving CSPs,
become tractable if 
the underlying hypergraphs have bounded 
$\hw$, $\ghw$, or, $\fhw$ and an appropriate decomposition is given. 
This gives rise to the 
problem of recognizing if a given CQ or CSP has $\hw$, $\ghw$, or, $\fhw$ bounded by 
some constant $k$. 
Formally, for {\it decomposition\/} $\in \{$HD, GHD, FHD$\}$ and $k \geq 1$, we consider the following family~of~problems:

\smallskip
\noindent
$\checkp{(\mathit{decomposition}, k)}$\\
\begin{tabular}{ll}
	\bf Input	& hypergraph $H = (V,E)$;\\
	\bf Output	& {\it decomposition\/}  of $H$ of width $\leq k$ if it 
	exists and answer `no' otherwise.
\end{tabular}
\smallskip

Clearly, bounded $\fhw$ defines the 
largest tractable class 
while bounded 
$\hw$ defines the smallest one. 
On the other hand, the problem
$\checkp{(\mathrm{HD},k)}$ is feasible in polynomial time~\cite{DBLP:journals/jcss/GottlobLS02}
while
the $\checkp{(\mathrm{GHD},k)}$~\cite{DBLP:journals/jacm/GottlobMS09} 
and 
$\checkp{(\mathrm{FHD},k)}$~\cite{DBLP:conf/pods/FischlGP18} 
problems 
are NP-complete even for $k = 2$.

Systems to solve
the $\checkp{(\mathrm{HD},k)}$ problem 
exist \cite{DBLP:journals/jea/GottlobS08,DBLP:journals/jcss/ScarcelloGL07}. 
In contrast, 
for $\checkp{(\mathrm{GHD},k)}$
and 
$\checkp{(\mathrm{FHD},k)}$, 
apart from exhaustive search over possible decomposition trees 
(which only works for small hypergraphs),  
no implementations have been reported yet \cite{DBLP:journals/tods/AbergerLTNOR17} -- with one exception: 
very recently, an interesting approach is presented in \cite{DBLP:conf/cp/FichteHLS18},
where SMT-solving is 
applied to the $\checkp{(\mathrm{FHD},k)}$ problem.
The same approach has been later extended to solve the $\checkp{(\mathrm{HD},k)}$ problem~\cite{DBLP:conf/alenex/SchidlerS20}.
In \cite{DBLP:journals/jea/GottlobS08}, 
tests of the 
$\checkp{(\mathrm{HD},k)}$ system 
are presented. However, 
a benchmark for systematically evaluating systems 
for the 
$\checkp{(\mathit{decomposition},k)}$ problem
with {\it decomposition\/} $\in \{$HD, GHD, FHD$\}$ and $k \geq 1$
were missing so far.
This motivates our first research goals.

\begin{myintro}
	\noindent{\bf Goal 1:} Create a comprehensive,  easily extensible benchmark of hypergraphs 
	corresponding to CQs or CSPs for the analysis of hypergraph decomposition algorithms.
\end{myintro}

\begin{myintro}
	\noindent{\bf Goal 2:} Use the
	benchmark from Goal 1 to find out if the hypertree width is, in general, small enough (say $\leq 5$) to allow for efficient evaluation of CQs of arbitrary arity and of CSPs. 
\end{myintro} 

Recently, in~\cite{DBLP:conf/pods/FischlGP18}, 
the authors have identified classes of CQs for which the 
$\checkp{(\mathrm{GHD},k)}$  and $\checkp{(\mathrm{FHD},k)}$ problems
become tractable
(from now on, we only speak about CQs; of course, all results apply equally to CSPs). To this end, the Bounded 
Intersection Property (BIP) and, more generally, the 
Bounded Multi-Intersection Property (BMIP)
have been introduced. 
The maximum number $d$ of attributes shared by two 
(resp.\ $c$) atoms is referred to as the 
intersection size (resp.\ $c$-multi-intersection size) of the CQ, which 
is similar to the notion of cutset width from the CSP literature~\cite{dechter2003}. 
We say that a class of CQs satisfies the BIP (resp.\ BMIP) 
if the number of attributes shared by two
(resp.\ by a constant number $c$ of) query atoms is bounded by 
some constant $d$. 

A related property is that of bounded degree, i.e., 
each attribute only occurs in a constant number of query atoms. Clearly, the BMIP
is an immediate consequence of bounded degree.
It has been shown in~\cite{DBLP:conf/pods/FischlGP18}
that $\checkp{(\mathrm{GHD},k)}$ is solvable in polynomial time for 
CQs whose underlying hypergraphs satisfy the BMIP. 
For CQs, the BMIP and bounded degree seem natural restrictions. For CSPs, the situation is not so clear. 
This yields the following research goals.

\begin{myintro}
	\noindent{\bf Goal 3:} Use the hypergraph benchmark from Goal~1 to 
	analyze how realistic the restrictions to low (multi-)inter\-section size,
	or low degree of CQs and CSPs are.
\end{myintro} 

\begin{myintro}
	\noindent{\bf Goal 4:} Verify 
	that for hypergraphs of low intersection size, 
	the $\checkp{(\mathrm{GHD},k)}$ problem
	indeed allows for efficient algorithms that work well in practice.
\end{myintro}

The tractability results for $\checkp{(\mathrm{FHD},k)}$ \cite{DBLP:conf/pods/FischlGP18}
are significantly weaker than for 
$\checkp{(\mathrm{GHD},k)}$: they involve a factor which is at least double-exponential 
in some ``constant'' (namely $k$, the bound $\delta$ on the degree and/or the bound $d$ on the intersection size). Hence, we want to investigate if (generalized) hypertree decompositions 
could be ``fractionally improved'' 
by taking the integral edge cover at each node in the 
HD or GHD and replacing it by a fractional edge cover. 
We will thus introduce the notion of {\em fractionally improved\/} HD 
which checks if there exists an HD of width $\leq k$, 
such that replacing each integral cover by a fractional cover yields an FHD of width
$\leq k'$ for given bounds $k,k'$ with $0 < k' < k$. 

\begin{myintro}
	\noindent{\bf Goal 5:} Explore the potential of fractionally improved HDs, i.e., 
	investigate if the improvements
	achieved are significant.
\end{myintro} 

In cases where $\checkp{(\mathrm{GHD},k)}$ and 
$\checkp{(\mathrm{FHD},k)}$ are intractable, we may have to settle for 
good approximations of $\ghw$ and $\fhw$. 
For GHDs, we may thus use the 
inequality  $\ghw(H) \leq 3 \cdot \hw(H) +1$, which 
holds for every 
hypergraph $H$~\cite{DBLP:journals/ejc/AdlerGG07}.
In contrast, for FHDs, the best known general, polynomial-time approximation is cubic. More precisely, 
in~\cite{DBLP:journals/talg/Marx10}, a polynomial-time algorithm is presented which, 
given a hypergraph $H$ with $\fhw(H) = k$,
computes an FHD of width $\mathcal{O}(k^3)$. 
In~\cite{DBLP:conf/pods/FischlGP18}, it is shown that a polynomial-time approximation up to a logarithmic factor is 
possible for any class of hypergraphs with bounded 
Vapnik--Chervonenkis dimension (VC-dimension; see
Section~\ref{sec:preliminaries} for a precise definition).
The problem of efficiently approximating the $\ghw$ 
and/or $\fhw$ 
leads us to the following goals. 

\begin{myintro}
	\noindent{\bf Goal 6:} Use the
	benchmark from Goal~1 to 
	analyze if, in practice, $\hw$ and $\ghw$ indeed differ by factor 3 or, if $\hw$ is typically much closer to $\ghw$
	than this worst-case bound.
\end{myintro} 

\begin{myintro}
	\noindent{\bf Goal 7:} Use the
	benchmark from Goal~1 to 
	analyze how realistic the restriction to small VC-dimen\-sion of 
	CQs and CSPs is.
\end{myintro} 

\smallskip
\noindent
{\bf Results.}
Our main results are as follows:

\smallskip
$\bullet$ \
We provide {\em HyperBench\/}, a comprehensive hypergraph benchmark of
initially over
3,000 hypergraphs (
see Section~\ref{sec:hyperbench}). 
This benchmark is exposed by a web interface, 
which allows the user to retrieve the hypergraphs or groups of hypergraphs together with a broad spectrum of properties of these hypergraphs, such as 
lower/upper bounds on $\hw$ and $\ghw$, (multi-)intersection size, degree, 
etc.

\smallskip

$\bullet$  \
We extend the software for HD computation from~\cite{DBLP:journals/jea/GottlobS08}
to also solve the 
$\checkp{(\mathrm{GHD},k)}$ problem. 
For a given hypergraph $H$, our system first computes the intersection size of $H$ and then applies the 
$\ghw$-algo\-rithm from~\cite{DBLP:conf/pods/FischlGP18},
which is parameterized by the
intersection size. We implement several improvements and 
we further extend the system to compute
also  ``fractionally improved''~HDs.

\smallskip

$\bullet$  \
We carry out an empirical analysis of the hypergraphs in the HyperBench benchmark. This analysis demonstrates, especially for real-world instances, 
that the restrictions to BIP, BMIP, bounded degree, and bounded VC-dimension are astonishingly realistic. Moreover,
on all hypergraphs in the HyperBench benchmark, we 
run our $\hw$- and $\ghw$-programs to identify 
(or at least bound) their $\hw$ and $\ghw$.
An interesting observation of our empirical study is that apart from the CQs 
also a significant portion of CSPs in our benchmark 
has small hypertree width (all non-random CQs have $\hw \leq 3$ and 
over 60\% of CSPs stemming from applications have 
$\hw \leq 5$). Moreover, for $\hw \leq 5$, 
in all of the 
cases where the $\ghw$-computation terminates, $\hw$ and $\ghw$ 
have identical values.

\smallskip

$\bullet$  \
In our study of the $\ghw$ of the hypergraphs in the HyperBench benchmark, 
we observe that a straightforward implementation of the algorithm from~\cite{DBLP:conf/pods/FischlGP18} 
for hypergraphs of 
low intersection size is too slow in many cases. We therefore present a new approach (based on so-called
``balanced separators'') with promising experimental results. It is interesting to note that
the new approach 
works particularly well in those situations which are particularly hard for the straightforward implementation,
namely hypergraphs $H$ where the test if $\ghw \leq k$ for given $k$ gives a ``no''-answer. 
Hence, combining the different approaches is very effective.

\smallskip

\noindent 
{\bf Structure.} This paper is structured as follows:
In Section~\ref{sec:related-work}, we summarize related work.
In Section~\ref{sec:preliminaries}, we recall some basic notions.  
In Section~\ref{sec:ghd-algs}, we describe our algorithms for solving the 
$\checkp{(\mathrm{GHD},k)}$ problem.
In Section~\ref{sec:hyperbench}, we present our system and test environment as well as our HyperBench 
benchmark.
This section also contains the first results of our empirical study of the hypergraphs in 
this benchmark.
In Section~\ref{sec:experiments} we report on the performance of the implementations of our GHD-algorithms as well as
a further extension of the system to allow for the computation of fractionally improved HDs.
We conclude in Section~\ref{sec:conclusion} by highlighting the most important lessons learned  from our empirical study and 
by identifying some appealing directions for future work.

\section{Related Work}
\label{sec:related-work}

We distinguish several types of works that are highly relevant to ours. 
The works most closely related are the descriptions of HD, GHD and FHD algorithms in~\cite{DBLP:journals/jcss/GottlobLS02,DBLP:conf/pods/FischlGP18} and the implementation of HD computation by the 
\detkdecomp\ 
program reported in~\cite{DBLP:journals/jea/GottlobS08}.
We have extended these works in several ways. Above all, we have incorporated our analysis tool 
(reported in 
Sections~\ref{sec:hyperbench}~and~\ref{sec:experiments}) and the GHD and FHD computations 
(reported in 
Sections~\ref{sec:ghd-algs}~and~\ref{sec:fhd-algs})
into the \detkdecomp\ program -- resulting in our \newdetkdecomp\ library, which is openly available on GitHub
at~\url{https://github.com/dmlongo/newdetkdecomp}. 
For the GHD computation, we have added heuristics to speed up the basic algorithm from~\cite{DBLP:conf/pods/FischlGP18}. Moreover, we have proposed a 
novel approach via balanced separators, which allowed us to significantly extend the range of instances for which 
the GHD computation terminates in reasonable time. 
We have also introduced a new form of decomposition method: the 
fractionally improved decompositions (see Section \ref{sec:fhd-algs}), which allow for a practical, lightweight form of FHDs.

The second important input to our work comes from the various sources 
\cite{
	DBLP:journals/pvldb/ArocenaGCM15%
	,BenediktCQs%
	,DBLP:conf/pods/BenediktKMMPST17%
	,berg2017maxsat%
	,DBLP:conf/icde/GeertsMPS14%
	,DBLP:journals/jea/GottlobS08%
	,DBLP:journals/pvldb/LeisGMBK015%
	,DBLP:conf/sigmod/JainMHHL16%
	,tpch}
which we took our CQs and CSPs from. 
Note that our main goal was not to add further CQs and/or CSPs to these benchmarks. 
Instead, we have aimed at taking and combining existing, openly accessible benchmarks of CQs and CSPs, 
and convert them into hypergraphs, which are then thoroughly analyzed. Finally, the hypergraphs and the analysis results are made openly accessible again. 

The third kind of works highly relevant to ours are previous analyses of CQs and CSPs. 
To the best of our knowledge, Ghionna et al.~\cite{DBLP:conf/icde/GhionnaGGS07} presented the first systematic study of HDs 
of benchmark CQs from TPC-H. However, Ghionna et al.\ 
pursued a research goal different  from ours in that they primarily wanted to find out to what extent HDs can actually 
speed up query evaluation. They achieved very positive results in this respect, which 
have recently been confirmed by the work of Perelman et al.\ \cite{DBLP:conf/sigmod/PerelmanR15}, Tu et al. \cite{DBLP:conf/sigmod/TuR15} and Aberger et al. \cite{DBLP:journals/tods/AbergerLTNOR17} on query evaluation using FHDs.
As a side result, Ghionna et al.\  also detected that CQs tend to have low hypertree width (a finding which was later confirmed in 
\cite{DBLP:journals/pvldb/BonifatiMT17,DBLP:conf/www/BonifatiMT19,DBLP:conf/sigmod/PicalausaV11} and also in our study).  
In a pioneering effort, Bonifati, Martens, and Timm \cite{DBLP:journals/pvldb/BonifatiMT17} have recently
analyzed an unprecedented, massive amount of  queries: they investigated 180,653,910 queries from 
(not openly available) query logs of several popular SPARQL endpoints. After elimination of duplicate queries, there were 
still 
56,164,661 queries left, out of which 26,157,880 queries were in fact CQs. 
The authors thus significantly extend previous work by 
Picalausa and Vansummeren \cite{DBLP:conf/sigmod/PicalausaV11}, who 
analyzed 3,130,177 SPARQL queries posed by humans and software robots at the DBPedia SPARQL endpoint. 
The focus in \cite{DBLP:conf/sigmod/PicalausaV11} is on structural properties of SPARQL queries such as keywords used and variable structure in 
optional patterns. There is one paragraph devoted to CQs, where it is noted that 99.99\% of ca.\  2 million CQs considered 
in \cite{DBLP:conf/sigmod/PicalausaV11} are acyclic. 

Many of the CQs  (over 15 million) analyzed in~\cite{DBLP:journals/pvldb/BonifatiMT17}
have arity 2 (here we consider the maximum arity  of all atoms in a CQ as the arity of the query), which means that all triples in such a SPARQL query have a constant at 
the predicate-position. 
Bonifati et al.\ made several interesting observations concerning the shape of these graph-like queries. For instance, they detected that
exactly one of these queries has $\tw = 3$, while all others have
$\tw \leq 2$ (and hence $\hw \leq 2$).
As far as the CQs of arity 3 are concerned (for CQs expressed as SPARQL queries, this is the maximum arity achievable), among many characteristics, also the hypertree width was computed by using the original $\detkdecomp$ program from~\cite{DBLP:journals/jea/GottlobS08}. 
Out of 6,959,510 CQs of arity 3, 
only 86 (i.e. 0.01\textperthousand) 
turned out to have $\hw = 2$ 
and 8 queries had $\hw = 3$
, while all other CQs of arity 3 are acyclic.  
Our analysis confirms that, also for non-random CQs of arity $> 3$, the hypertree width indeed tends to be low, with the majority of queries being even acyclic.

Bonifati et al. continued on this line of work and analyzed in~\cite{DBLP:conf/www/BonifatiMT19} a yet bigger collection of SPARQL queries coming from Wikidata.
This repository of 208,215,209 queries was first made available by Malyshev et al. for the work in~\cite{DBLP:conf/semweb/MalyshevKGGB18}.
Bonifati et al. divided this dataset into four disjoint sets:
queries for	which the HTTP request was successful, further partitioned into organic and robotic queries;
and timeout queries, again further partitioned into organic and robotic queries.
Organic queries are the ones classified by Malyshev et al. as posed by humans, while robotic queries have been classified as produced by synthetic algorithms.
While Malyshev et al. analyzed only the sets of successful queries, Bonifati et al. extended their study to timeout queries.
The latter turned out to be the most interesting w.r.t. structural analysis.
As in~\cite{DBLP:journals/pvldb/BonifatiMT17}, the focus of~\cite{DBLP:conf/www/BonifatiMT19} is on examining property paths of SPARQL queries and, in particular, having a clear picture of the structural characteristics of recursive properties.

In~\cite{DBLP:conf/www/BonifatiMT19}, structural analysis builds upon conjunctive queries and variants thereof.
In total, Bonifati et al. identify 176,679,495 robotic and 342,576 organic queries for the largest fragment of CQs (namely C2RPQ+), which constitutes circa 85\% of the whole dataset.
Structural analysis showed that these queries are \emph{mildly cyclic}, i.e., their treewidth is bounded by a small constant.
In particular, for all these queries $\tw \leq 4$ holds.
For a different fragment of CQs (namely $\mathrm{CQ_{OF+}}$), a structural analysis based on hypergraphs is more suitable, thus they computed $\hw$ for a total of 1,915,550 $\mathrm{CQ_{OF+}}$ queries.
It turns out that 590,005 queries have $\hw = 2$, while the rest has $\hw = 1$, i.e., they are acyclic.

For the analysis of CSPs, much less work has been done. 
Although it has been shown that exploiting {(hyper-)}tree
decompositions may significantly improve the performance of CSP solving 
\cite{DBLP:journals/aicom/AmrounHA16,DBLP:journals/jetai/HabbasAS15,DBLP:conf/sara/KarakashianWC11,DBLP:conf/aiia/LalouHA09}, 
a systematic study on the (generalized) hypertree width of CSP instances has only been carried out by few works \cite{DBLP:journals/jea/GottlobS08,DBLP:conf/aiia/LalouHA09,Schafhauser06}. To the best of our knowledge, we are the first to analyze the \hw, \ghw, and \fhw\ of ca.\ 2,000 CSP instances, where most of these instances have not been studied in this respect before.

It should be noted that the focus of our work is different from the above mentioned previous works: above all, we wanted to test the 
practical feasibility of various algorithms for HD, GHD, and FHD computation (including both, previously presented algorithms and new ones developed as part of this work). As far as our repository of hypergraphs (obtained from CQs and CSPs) is concerned, we emphasize open accessibility. 
Thus, users can analyze their CQs and CSPs (with our implementations of HD, GHD, and FHD algorithms) or they can analyze 
new decomposition algorithms (with our hypergraphs, which cover quite a broad range of characteristics). 
In fact, in the recent works on HD and FHD computation via SMT solving~\cite{DBLP:conf/cp/FichteHLS18,DBLP:conf/alenex/SchidlerS20}, 
the HyperBench benchmark has 
already been used for 
the experimental evaluation. 
In~\cite{DBLP:conf/cp/FichteHLS18} a novel approach to $\fhw$ computation via an efficient encoding of the 
check-problem for FHDs to SMT (SAT modulo Theory) is presented. The tests were carried out with 2,191 hypergraphs from 
the initial version of the HyperBench.
For all of these hypergraphs we have established at least some upper bound on the $\fhw$ either by our $\hw$-computation
or by one of our new algorithms presented in Sections~\ref{sec:ghd-algs}~and~\ref{sec:fhd-algs}.
In contrast, the exact algorithm in~\cite{DBLP:conf/cp/FichteHLS18} found FHDs only for 1.449 instances (66\%). 
In 852 cases, both our algorithms and the algorithm in~\cite{DBLP:conf/cp/FichteHLS18} found FHDs of the same width;
in 560 cases, an FHD of lower width was found in \cite{DBLP:conf/cp/FichteHLS18}.
By using the same benchmark for the tests,
the results  in~\cite{DBLP:conf/cp/FichteHLS18} and ours 
are comparable and have thus provided 
valuable input for future improvements of the algorithms by 
combining the 
different strengths and weaknesses of the two 
approaches.

The use of the same benchmark has also allowed us to provide feedback to the authors of~\cite{DBLP:conf/cp/FichteHLS18} for debugging their system: in 9 out of 2,191  cases, the ``optimal'' value for the $\fhw$ computed in~\cite{DBLP:conf/cp/FichteHLS18} was apparently erroneous, since it was higher than the $\hw$ found out by our analysis; note that upper bounds on 
the width are, in general, more reliable than lower bounds since it is easy to verify if a given decomposition indeed has the desired properties, whereas ruling out the existence of a decomposition of a certain width is a complex and error-prone task.

The work in~\cite{DBLP:conf/ijcai/GottlobOP20} represents a follow-up of this paper and the conference version~\cite{DBLP:conf/pods/FischlGLP19}.
The main goal of~\cite{DBLP:conf/ijcai/GottlobOP20} consists in providing major improvements in computing GHDs.
To this aim the authors present a parallel algorithm for computing GHDs 
based on the balanced separator method described in this paper,
and a hybrid approach that combines the parallel and sequential decomposition algorithms.
More specifically, the hybrid approach uses a parallel version of the balanced separator algorithm to split a large hypergraph
into smaller components and then uses the sequential version of $\detkdecomp$ 
from~\cite{DBLP:journals/jea/GottlobS08}
to quickly decompose them.
Moreover, the authors of ~\cite{DBLP:conf/ijcai/GottlobOP20} propose new methods to simplify the input hypergraph 
and apply new heuristics to reduce the search space.
Our HyperBench benchmark was fundamental for the experimental evaluation of all these techniques.
Indeed, the version of HyperBench presented in~\cite{DBLP:conf/pods/FischlGLP19} was used as a baseline to show that
the newly proposed algorithms can be used to efficiently compute GHDs on modern machines for a wide range of CSP instances.

Similarly to our analysis of properties that make $\checkp{(\mathit{decomposition},k)}$ tractable or easy to approximate,
HyperBench has been used to empirically test the validity of theoretical hypotheses.
In~\cite{korhonen2019potential}, the edge clique cover size of a graph is identified as a parameter allowing fixed-parameter-tractable algorithms for enumerating potential maximal cliques.
The latter can be used to compute exact $\ghw$ and $\fhw$.
An edge clique cover of a graph is a set of cliques of the graph that covers all of its edges.
In case of a CSP with $n$ variables and $m$ constraints, the set of constraints is an edge clique cover of the underlying (hyper)graph.
Thus, this property can be exploited for CSPs having $n > m$ and HyperBench has been used to verify that it happens in circa 23\% of the instances.

\section{Preliminaries}
\label{sec:preliminaries}

\subsection{CQs, CSPs and Hypergraphs}
We treat \emph{conjunctive queries} (CQs) and \emph{constraint satisfaction problems} (CSPs) as first-order formulae using only connectives in $\{\exists, \land\}$ and disallowing $\{\forall,\vee,\neg\}$.

A \emph{hypergraph} $H = (V(H), E(H))$ is a pair consisting of a set of vertices $V(H)$ and a set of non-empty (hyper)edges $E(H) \subseteq 2^{V(H)}$.
We assume w.l.o.g.\  that there are no isolated vertices, i.e., for each $v \in V(H)$, there is at least one edge $e \in E(H)$ such that $v \in e$.
We can thus identify a hypergraph $H$ with its set of edges $E(H)$ with the understanding that 
$V(H) = \{ v \in e \mid e \in E(H) \} $.  A \emph{subhypergraph} $H'$ of $H$ is then simply a subset of (the edges of) $H$.

Given a formula $\phi$ corresponding to either a CQ or a CSP, the hypergraph $H_\phi$ corresponding to $\phi$ has $V(H_\phi) = \mathrm{vars}(\phi)$ and, for each atom $a$ of $\phi$, $\mathrm{vars}(a) \in E(H_\phi)$.

We are frequently dealing with sets of sets of vertices (e.g., sets of edges). For $S \subseteq 2^{V(H)}$, we write 
$\bigcup S$ and $\bigcap S$ as a short-hand for taking the union or intersection, respectively, of this set of sets of vertices, i.e., 
for $S = \{s_1, \dots, s_\ell\}$, we have $\bigcup S = \bigcup_{i=1}^\ell s_i$ and 
$\bigcap S = \bigcap_{i=1}^\ell s_i$.

\subsection{Hypergraph Decompositions and Width Measures}
We consider here three notions of hypergraph decompositions with their associated notions of width.
To this end, we first introduce the notion of (fractional) edge covers.

Let $H = (V(H),E(H))$ be a hypergraph and consider an \emph{edge weight function} $\gamma \colon E(H) \to [0,1]$.
We define  the set $B(\gamma)$ of all vertices covered by $\gamma$ and the weight of $\gamma$ as 
\begin{eqnarray*}
	B(\gamma) & = &\left\{ v\in V(H) \mid \sum_{e\in E(H), v\in e} \gamma(e) \geq 1 
	\right\}, \\
	\ \mathit{weight}(\gamma) & =& \sum_{e \in E(H)} \gamma(e).
\end{eqnarray*}
We call $\gamma$ a \emph{fractional edge cover} of a set $X \in V(H)$ by edges in $E(H)$, if $X \subseteq B(\gamma)$.
We also consider an \emph{integral edge cover} as a function $\lambda \colon E(H) \to \{0,1\}$, i.e., a fractional edge cover whose values are restricted to $\{0,1\}$ values. Following \cite{DBLP:journals/jcss/GottlobLS02}, we can also treat $\lambda$ as a set with $\lambda \subseteq E(H)$  (namely, the set of edges $e$ with $\lambda(e) = 1$) and the weight as the cardinality of such a set of edges.
In the following, to emphasize the nature of the function we are dealing with, we will use $\gamma$ for fractional edge covers and $\lambda$ for integral edge covers.  

We now introduce some relevant notions of decompositions.

A tuple $\langle T, (B_u)_{u \in T} \rangle$ is a \emph{tree  decomposition} (TD) of hypergraph $H =(V(H),E(H))$,
if $T = (N(T),E(T))$ is a tree, every $B_u$ is a subset of $V(H)$, and the following conditions are satisfied:
\begin{enumerate}
	\item[(1)] 
	For every edge $e \in E(H)$, there is a node $u$ in $T$, such that  $e \subseteq B_u$, and
	\item[(2)] for every vertex $v \in V(H)$,  $\{u \in T \mid v \in B_u\}$ is connected in $T$.
\end{enumerate}
The vertex sets $B_u$ are usually referred to as the \emph{bags} of the TD. 
By slight abuse of notation, we write $u \in T$ to express that $u$ is a node in $N(T)$. Condition (2) is also called 
the ``connectedness condition".

We use the following notational conventions throughout
this paper. To avoid confusion, we will consequently refer
to the elements in $V(H)$ as {\em vertices\/} of the hypergraph and to
the elements in $N(T)$ as the {\em nodes\/} of the decomposition. For
a node $u \in T$, we write $T_u$ to denote the subtree of $T$ rooted
at $u$. By slight abuse of notation, we will often write $u' \in T_u$
to denote that $u'$ is a node in the subtree $T_u$ of $T$.
Finally, we define $V(T_u) = \bigcup_{u' \in T_u} B_{u'}$.

A \emph{fractional  hypertree decomposition} (FHD) of a hypergraph  $H=(V(H),E(H))$ is a tuple\\ 
$\left< T, (B_u)_{u\in T}, (\gamma_u)_{u\in T} \right>$, such that 
$\left< T, (B_u)_{u\in T}\right>$ is a TD of $H$ and the following condition holds:
\begin{enumerate}
	\item[(3)] For each $u\in T$,  $B_u \subseteq  B(\gamma_u)$ holds, i.e., $\gamma_u$ is a fractional 
	edge cover of $B_u$.
\end{enumerate}
Following our notational convention, a \emph{generalized  hypertree decomposition} (GHD) is an FHD, where $\lambda_u$ is an integral edge weight function for every $u \in T$.
Hence, by condition (3), $\lambda_u$ is an integral edge cover of $B_u$.
A \emph{hypertree decomposition} (HD) of $H$ is a GHD with the following additional condition (referred to as the ``special condition'' in~\cite{DBLP:journals/jcss/GottlobLS02}):
\begin{enumerate}
	\item[(4)] For each $u\in T$, $ V(T_u) \cap B(\gamma_u) \subseteq B_u$, where 
	$V(T_u)$ denotes the union of all bags in the subtree of $T$ rooted at $u$.
\end{enumerate}
Because of condition (4), it is important to consider $T$ as a \emph{rooted} tree in case of HDs. 
For TDs, FHDs, and GHDs,  the root of $T$ can be arbitrarily chosen or simply ignored.
The \emph{width} of an FHD, GHD, or HD is defined as 
the maximum weight of the functions $\gamma_u$  over all nodes $u \in T$.
The fractional hypertree width,  generalized hypertree width, and
hypertree width of $H$ (denoted $\fhw(H)$, 
$\ghw(H)$, and $\hw(H)$) is the minimum width over all FHDs, GHDs, and HDs~of~$H$.

\subsection{Components and Separators}
For a set $U \subseteq V(H)$ of vertices, we define $[U]$-components of a hypergraph $H$ as follows:
\begin{itemize}
	\item We define {\em $[U]$-adjacency\/} as a binary relation on $E(H)$ as follows: 
	two edges $e_1,e_2 \in E(H)$ are {\em $[U]$-adjacent\/}, if $(e_1 \cap e_2) \setminus U \neq \emptyset$ holds.
	\item We define {\em $[U]$-connectedness\/} as the transitive closure of the {\em $[U]$-adjacency\/} relation.
	\item A {\em $[U]$-component\/} of $H$ is a maximally $[U]$-connected subset $C \subseteq E(H)$ .
\end{itemize}

For a set of edges $S \subseteq E(H)$, we say that $C$ is 
``\emph{$[S]$-connected}'' or 
an ``\emph{$[S]$-component\/}'' as a short-cut for 
$C$ is ``$[W]$-connected'' or a ``$[W]$-component'', respectively, 
with $W = \bigcup_{e \in S} e$. We also call $S$ a \emph{separator} in this context.
The \emph{size of an $[S]$-component} $C$ is defined as the number of edges in the component.
For a hypergraph $H$ and a set of edges $S \subseteq E(H)$, we say that $S$ is a 
\emph{balanced separator} if all $[S]$-components of $H$ have a size 
$\leq \frac{\lvert E(H) \rvert}{2}$.
It was shown in~\cite{DBLP:journals/ejc/AdlerGG07} that, for every GHD $\langle T, (B_u)_{u \in T}, (\lambda_u)_{u \in T} \rangle$ of a hypergraph $H$, 
there exists a node $u \in T$ such that $\lambda_u$ is a balanced separator of $H$.
This property can be made use of when searching for a GHD of size $k$ of a 
hypergraph $H$ as we shall show  in 
Section~\ref{sec:balsep} below.

\subsection{Computing Hypertree Decompositions (HDs)}
\label{sect:ComputingHDs}

We briefly recall the basic principles of the \detkdecomp program from~\cite{DBLP:journals/jea/GottlobS08} for computing HDs.
It is relevant for our work in that it is the first implementation of the original top-down nondeterministic HD algorithm presented in~\cite{DBLP:journals/jcss/GottlobLS02}.
Even though the first implementation of a deterministic HD algorithm is \optkdecomp~\cite{DBLP:conf/dexa/GottlobLS99}, it is based on a different characterization of $\hw$, which is not suitable for our purposes.

For fixed $k \geq 1$, \detkdecomp tries to construct an HD of a hypergraph $H$ 
in a top-down
manner. It thus maintains a set $C$ of edges, which is initialized to $C = E(H)$. 
For a node $u$ in the HD (initially, this is the root of the HD), 
it ``guesses'' an edge cover $\lambda_u$, i.e., $\lambda_u \subseteq E(H)$ and 
$\lvert \lambda_u \rvert \leq k$. 
For fixed $k$, there are only polynomially many possible values $\lambda_u$. 
\detkdecomp then proceeds by determining all $[\lambda_u]$-components $C_i$ with $C_i \subseteq C$. 
The special condition of HDs restricts the possible choices for 
$B_u$ and thus guarantees that the $[\lambda_u]$-components inside $C$ 
and the $[B_u]$-components inside $C$ coincide. 
This is the crucial property for ensuring polynomial time complexity
of HD-computation -- at the price of possibly missing GHDs with a lower width. 

Now let $C_1, \dots, C_\ell$ denote the $[\lambda_u]$-components with $C_i \subseteq C$.
By the maximality of components, these sets $C_i$ are pairwise disjoint.
Moreover, it was shown in~\cite{DBLP:journals/jcss/GottlobLS02} that 
if $H$ has an HD of width $\leq k$, then it also has an HD of width 
$\leq k$ such that the edges in each $C_i$ are ``covered'' in different subtrees below $u$. 
More precisely, this means that $u$ has $\ell$ child nodes $u_1, \dots, u_\ell$, such that 
for every $i$ and every $e \in C_i$, there exists a node $u_e$ in the subtree rooted at $u_i$
with $e \subseteq B_{u_e}$. Hence, \detkdecomp recursively searches for an HD of the 
hypergraphs $H_i$ with $E(H_i) = C_i$ and $V(H_i) = \bigcup C_i$ with the slight extra feature that 
also edges from $E(H) \setminus C_i$ are allowed to be used in the $\lambda$-labels of these HDs.

\subsection{Favorable Properties of Hypergraphs}
We are interested in certain structural properties of hypergraphs that make the $\checkp{(\mathrm{GHD},k)}$ and $\checkp{(\mathrm{FHD},k)}$ problems tractable or efficient to approximate for large classes of hypergraphs.
We refer to the terminology of~\cite{DBLP:journals/corr/abs-2002-05239} as it makes uniform the one originally introduced in~\cite{DBLP:conf/pods/FischlGP18}.

\begin{myDefinition}
	For $c \geq 1$, $d \geq 0$, a hypergraph $H = (V(H),E(H))$ is a \emph{$(c,d)$-hypergraph} if the intersection of any $c$ edges in $E(H)$ has at most $d$ elements, i.e., for every subset $E' \subseteq E(H)$ with $\lvert E' \rvert = c$, we have $\lvert \bigcap E' \rvert \leq d$. 
\end{myDefinition}

\begin{myDefinition}
	A hypergraph $H = (V(H),E(H))$ has \emph{$c$-multi-intersection size} $d$ if $H$ is a $(c,d)$-hypergraph. 
	In the special case of $c=2$, we speak of intersection size of $H$, while if we do not have a particular $c$ in mind, we simply speak of multi-intersection size of $H$.
\end{myDefinition}

\begin{myDefinition}
	\label{def:bip}
	A class $\mathcal{C}$ of hypergraphs satisfies the \emph{bounded multi-intersection property (BMIP)}, if there exist $c \geq 1$ and $d \geq 0$, such that every $H$ in $\mathcal{C}$ is a $(c,d)$-hypergraph.
	As a special case, $\mathcal{C}$ satisfies the \emph{bounded intersection property (BIP)}, if there exists $d \geq 0$, such that every $H$ in $\mathcal{C}$ is a $(2,d)$-hypergraph.
\end{myDefinition}

There are further relevant properties of (classes of)
hypergraphs: bounded degree and bounded Vapnik–Chervonenkis
dimension (VC-dimension).

\begin{myDefinition}
	The \emph{degree} $\mathrm{deg}(H)$ of a hypergraph $H$ is defined as he maximum number $\delta$ of hyperedges in which a vertex occurs, i.e., $\delta = \max_{v \in V(H)} \lvert \{e \in E(H) \mid v \in E(H)\} \rvert$.
	A class $\mathcal{C}$ of hypergraphs has bounded degree, if there exists $\delta \geq 1$, such that every hypergraph $H$ in $\mathcal{C}$ has degree $\leq \delta$. 
\end{myDefinition}

\begin{myDefinition}[\cite{Vapnik1971}]
	\label{def:vc}
	Let $H=(V(H),E(H))$ be a hypergraph, and $X\subseteq V(H)$ a set of vertices. 
	Denote by 
	$E(H)|_X =\{X \cap e\, |\, e\in E(H)\}$. $X$ is called {\em shattered} if 
	$E(H)|_X=2^X$.
	The {\em Vapnik-Chervonenkis dimension (VC dimension)} of $H$ is 
	the maximum cardinality of a shattered subset of $V$. 
	We say that a class $\mathcal{C}$ of hypergraphs  has bounded VC-dimension, if there exists $v \geq 1$, 
	such that every hypergraph $H\in \mathcal{C}$ has VC-dimension $\leq v$.
\end{myDefinition}

Note that a hypergraph $H$ with degree bounded by $\delta$ is a $(\delta+1,0)$-hypergraph.
Thus, bounded degree implies the BMIP, which in turn implies bounded VC-dimension~\cite{DBLP:conf/pods/FischlGP18}.

The aforementioned properties help to solve or approximate the problems $\checkp{(\mathrm{GHD},k)}$ and $\checkp{(\mathrm{FHD},k)}$:
\begin{myTheorem}[\cite{DBLP:conf/pods/FischlGP18,DBLP:journals/corr/abs-2002-05239}]
	\label{th:tractability}
	Let $\mathcal{C}$ be a class of hypergraphs. 
	\begin{itemize}
		\item If $\mathcal{C}$ has the BMIP, then the $\checkp{(\mathrm{GHD},k)}$ problem is solvable in polynomial time for arbitrary $k \geq 1$. 
		Consequently, this tractability holds if $\mathcal{C}$ has bounded degree or the BIP (which each imply the BMIP)~\cite{DBLP:conf/pods/FischlGP18}.
		
		\item If $\mathcal{C}$ has bounded degree or the BIP, then the $\checkp{(\mathrm{FHD},k)}$ problem is solvable in polynomial time for arbitrary 
		$k \geq 1$~\cite{DBLP:conf/pods/FischlGP18,DBLP:journals/corr/abs-2002-05239}.
		
		\item If $\mathcal{C}$ has bounded VC-dimension, then the $\fhw$ can be approximated in polynomial time up to a 
		log-factor~\cite{DBLP:conf/pods/FischlGP18}.
	\end{itemize}
\end{myTheorem}

\section{GHD Algorithms}
\label{sec:ghd-algs}
In this section, we present two different ways to implement the tractable algorithm for computing GHDs in case of 
BIP~\cite{DBLP:conf/pods/FischlGP18}. We shall refer to these two implementations as \globalbip{} and \localbip{}.
Moreover, we shall present a completely new approach to computing GHDs, which will be referred to as \balsep{}.
These variants exploit low intersection size to compute a GHD of width $\leq k$ of a hypergraph in polynomial time.
While \globalbip{} and \localbip{} differ in the time of the computation of certain sets of subedges, the \balsep{} algorithm is a novel approach in computing decompositions by means of \emph{balanced separators}.

\subsection{Theoretical Background}
In~\cite{DBLP:conf/pods/FischlGP18}, the $\checkp{(\mathrm{GHD},k)}$ problem was proved to be intractable for 
$k \geq 2$, but the authors also identified tractable fragments corresponding to classes of hypergraphs having bounded (multi-)intersection size.
In the following we focus on the intersection size, i.e., the maximum intersection size of any two edges in a hypergraph and explain how this property can be used for tractable $\ghw$ computation.
Since the three variants presented later rely on the same central idea, we refer to the theoretical algorithm as the \emph{$ghw$-algorithm}.

For a given hypergraph $H = (V(H), E(H))$, the $\ghw$-algorithm adds a polynomial-time computable set $f(H,k)$ of subedges of edges in $E(H)$ to the hypergraph $H$.
The set $f(H,k)$ has the property that $\ghw(H) = k$ if and only if $\hw(H') = k$, where $H' = (V(H), E(H) \cup f(H,k))$.
Thus, it is possible to use $\checkp{(\mathrm{HD},k)}$ to solve $\checkp{(\mathrm{GHD},k)}$ and achieve tractability.
The particular set
\begin{equation}
f(H,k) = \bigcup_{e \in E(H)}\Bigg(
\bigcup_{e_1,\dots,e_j \in (E(H) \setminus \{e\}),\, j \leq k} 
2^{(e \cap (e_1 \cup \dots \cup e_j))}
\Bigg)
\label{eq:bip-subsets}
\end{equation}
contains, for each $e \in E(H)$, all subsets of intersections of $e$ with up to $k$ edges of $H$ different from $e$.
Although $f(H,k)$ could in general contain an exponential number of elements, for fixed $k$ and intersection size of $H$ bounded by $d$, the set $e \cap (e_1 \cup \dots \cup e_j)$ contains at most $d \cdot k$ elements and, therefore, $\lvert f(H,k) \rvert$ is polynomially bounded.

Since the $\ghw$-algorithm relies on the computation of HDs, we implemented a new version of the $\hw$-algorithm in~\cite{DBLP:journals/jea/GottlobS08} and called it \newdetkdecomp{}.
Our program contains also the implementations of all the algorithms in this section, but we defer a detailed discussion of the complete library to Section~\ref{sec:experiments}.

\subsection{The \globalbip{} Algorithm}
Given $H$ and $k \geq 1$, a straightforward implementation of the $\ghw$-algorithm consists in computing the set $f(H,k)$, creating the hypergraph $H' = (V(H), E(H) \cup f(H,k))$, and finally computing an HD of $H$ of width $\leq k$, if it exists.
Since here we compute the set $f(H,k)$ a priori \emph{globally} for the whole hypergraph, we call this algorithm \globalbip{}.

Algorithm~\ref{alg:globalbip} is a detailed description of \globalbip{}.
The input consists of a hypergraph $H$ and a constant $k$, while the output is a GHD of $H$ of width $\leq k$, if it exists, and \textsc{null} otherwise.
In line 2, we first compute $f(H,k)$ as in Equation~\ref{eq:bip-subsets} and then, in line 3, we create the hypergraph $H'$ which is obtained by adding the subedges in $f(H,k)$ to $H$.
In line 4, we call \newdetkdecomp on $H'$ and $k$ as a black box and store its output in the variable $\mathit{HD}$.
If $\mathit{HD}$ is \textsc{null}, a GHD of width $\leq k$ of $H$ does not exist, therefore we return \textsc{null}.
Otherwise, we need to ``fix" the decomposition as described in lines 6-10.
In particular, for each node $u$ of $\mathit{HD}$, and for each edge of $f(H,k)$ in $\lambda_u$, i.e., $e' \in (\lambda_u \cap f(H,k))$ that is not an edge in the original hypergraph $H$, we substitute $e$ with an edge $e' \in E(H)$ such that $e \subseteq e'$.
In this way we obtain a new edge cover $\lambda_u'$ such that $B_u \subseteq B(\lambda_u) \subseteq B(\lambda_u')$, but also $\lvert \lambda_u \rvert = \lvert \lambda_u' \rvert$.
Thus, the new decomposition still satisfies all the properties of a GHD and its width is still $\leq k$.
Eventually, in line 11, we return $\mathit{HD}$.

\begin{algorithm}
	\SetKwInOut{KwPar}{Parameter}
	\SetKwFunction{NewDetKDecomp}{NewDetKDecomp}
	
	\DontPrintSemicolon
	\KwIn{A hypergraph $H$.}
	\KwPar{An integer $k \geq 1$.}
	\KwOut{A GHD of $H$ of width $\leq k$ if it exists, \textsc{null} otherwise.}
	
	\Begin{
		$f(H,k) \leftarrow $ compute as in Equation~\ref{eq:bip-subsets}\;
		$H' \leftarrow (V(H), E(H) \cup f(H,k))$\;
		$\mathit{HD} \leftarrow$ \NewDetKDecomp{$H', k$}\;
		\If{$\mathit{HD} \neq$ \textsc{null}}{
			\ForEach{$u \in \mathit{HD}$}{
				\ForEach{$e \in (\lambda_u \cap f(H,k))$}{
					$e' \leftarrow e' \in E(H)$ such that $e \subseteq e'$\;
					$\lambda_u(e) \leftarrow 0$\;
					$\lambda_u(e') \leftarrow 1$\;
				}
			}
		}
		\Return{$\mathit{HD}$}
	}
	
	\caption{\globalbip{}}
	\label{alg:globalbip}
\end{algorithm}

\subsection{The \localbip{} Algorithm}
\label{sec:localbip}
The main drawback of \globalbip{} is that the size of $f(H,k)$, though polynomial, could be huge for practical purposes.
Therefore we looked at ways to reduce the number of edges to add to $H$ by restricting the computation only to those edges that might be actually necessary.
The approach we used follows from an observation about the role played by $f(H,k)$ in the tractability proof in~\cite{DBLP:conf/pods/FischlGP18}.

Recall that the proof uses $\checkp{(\mathrm{HD},k)}$ on the hypergraph $H'$ to answer $\checkp{(\mathrm{GHD},k)}$ for the hypergraph $H$.
To do this in a sound way, the set $f(H,k)$ has to contain all the edges that could be used to cover possible bags of $H'$ in an HD without changing the width.
Consider a top-down construction of a GHD of $H$.
At some point we might want to choose, for some node $u$, a bag $B_u$ such that $x \notin B_u$ for some variable $x \in B(\lambda_u) \cap V(T_u)$.
This choice would violate condition (4) of HDs and would not be allowed for the computation of an HD.
In particular, there is an edge $e$ with $x \in e$ and $\lambda_u(e) = 1$.
For this reason, the set $f(H,k)$ contains an edge $e'$ such that $e' \subset e$ and $x \notin e'$.
Hence, we can substitute $e$ with $e'$ in the cover $\lambda_u$ (i.e., $\lambda_u(e)=0,\, \lambda_u(e')=1$) to eliminate the violation of condition (4).
Moreover, because of the connectedness condition, there is no need to look at the intersection of $e$ with arbitrary edges in $E(H)$, instead we consider only the intersections of $e$ with unions of edges that may possibly occur in bags of $T_u$.
In other words, for each node $u$ of the decomposition, we consider only an appropriate subset $f_u(H,k) \subseteq f(H,k)$.
More specifically, for the current node $u$, let $H_u \subseteq H$ be the component we want to decompose.
Then, we define $f_u(H,k)$ as follows:
\begin{equation}
f_u(H,k) = \bigcup_{e \in E(H)}\Bigg(
\bigcup_{e_1,\dots,e_j \in (E(H_u) \setminus \{e\}),\, j \leq k} 
2^{(e \cap (e_1 \cup \dots \cup e_j))}
\Bigg)
\label{eq:bip-local-subsets}
\end{equation}

We call the resulting algorithm \localbip{} because the set of edges $f_u(H,k)$ is computed \emph{locally} for each node $u$ during the construction of the decomposition.
It follows \newdetkdecomp closely, but it differs in the search of the separators.
In particular, while decomposing $H$, the algorithm first tries all possible $\ell$-combinations of edges in $E(H)$ (with $\ell \leq k$) and only if the search does not succeed, it tries $\ell$-combinations of subedges in $f_u(H,k)$.

\subsection{The \balsep{} Algorithm}
\label{sec:balsep}
So far we have presented two adaptations of the theoretical $\ghw$-algorithm from~\cite{DBLP:conf/pods/FischlGP18}.
On the one hand, they extend \newdetkdecomp to compute GHDs and exploit bounded intersection size for tractability.
On the other hand, they do not introduce any significant algorithmic innovation.
In the following we describe a novel approach to compute GHDs that makes use of sets of edges called \emph{balanced separators}.
We first extend the terminology of Section~\ref{sec:preliminaries}, then give a detailed description of the algorithm, and finally prove that our algorithm is sound and complete.

\subsubsection{Balanced Separators and Special Edges}
Recall that a \emph{hypergraph\/} is a pair $ H = (V(H), E(H)) $, consisting of a set $ V(H) $ of \emph{vertices} and a 
set $ E(H) $ of hyperedges (or, simply \emph{edges}), which are non-empty subsets of $ V(H) $.
Since we assume that hypergraphs do not have isolated vertices, we can identify a hypergraph $H$ with its set of edges $E(H)$.
Then, a subhypergraph $H'$ of $H$ is a subset of (the edges of) $H$.

Starting off with a hypergraph $H$, the \balsep{} algorithm has to deal with subhypergraphs $H' \subseteq H$ 
augmented by a set $S_p$ of special edges. A special edge is simply a set of vertices from $H$. Intuitively, special edges correspond to bags $B_u$ in a GHD of $H$.
Thus, an {\em extended subhypergraph\/} of $H$ is of the form  
$H' \cup S_p$,
where $H' \subseteq H$ is a subhypergraph and $S_p$ is a set of special edges.

We now extend three crucial definitions from hypergraphs to extended subhypergraphs, namely components, 
balanced separators and GHDs. 
We recall that, even though a separator is a set of vertices, it can be defined as a set of edges.
Then, for $S \subseteq E(H)$, an $[S]$-component  is a $[W]$-component with $W = \bigcup_{e \in S} e$.
In our algorithm we use the fact that, because of the edges in $f(H,K)$, 
it is always possible to choose a separator $\lambda_u$ such that $B(\lambda_u) = B_u$~\cite{DBLP:conf/pods/FischlGP18}.
Hence, there will not be any need to distinguish between vertex and edge separators.
We start with components.

\begin{myDefinition}[components of extended subhypergraphs]
	For a set $U \subseteq V(H)$ of vertices, we define $[U]$-components of an extended subhypergraph $H' \cup S_p$ of $H$ as follows:
	\begin{itemize}
		\item We define {\em $[U]$-adjacency\/} as a binary relation on $H' \cup S_p$ as follows: 
		two (possibly special) edges $f_1,f_2 \in H' \cup S_p$ are {\em $[U]$-adjacent\/}, if 
		$(f_1 \cap f_2) \setminus U \neq \emptyset$ holds.
		\item We define {\em $[U]$-connectedness\/} as the transitive closure of the {\em $[U]$-adjacency\/} relation.
		\item A {\em $[U]$-component\/} of $H' \cup S_p$ is a maximally
		$[U]$-connected subset $C \subseteq H' \cup S_p$ .
	\end{itemize}
\end{myDefinition}

Hence, if $C_1, \dots, C_\ell$ are the $[U]$-components of $H' \cup S_p$, then 
$H' \cup S_p$ is partitioned into $C_0 \cup C_1 \cup  \dots \cup  C_\ell$, such that 
$C_0 = \{ f \in H' \cup S_p \mid f \subseteq U \}$.

We next define balanced separators. While we give a definition w.r.t. sets of vertices $U \subseteq V(H)$,
they can be alternatively defined in terms of sets of edges $S \subseteq E(H)$.

\begin{myDefinition}[balanced separators]
	Let $H' \cup S_p$ be an extended subhypergraph of a hypergraph $H$
	and let  $U \subseteq V(H)$ be a set of vertices of $H$. 
	The set $U$ is a \emph{balanced separator} of $H' \cup S_p$ 
	if for each $[U]$-component $C_i$ of $H' \cup S_p$, 
	$\lvert C_i \rvert  \leq \frac{\lvert H' \cup S_p \rvert}{2} $
	holds.  In other words, no $[U]$-component must contain more than half the edges of $H' \cup S_p$. 
\end{myDefinition}

Finally, we extend GHDs to extended subhypergraphs.

\begin{myDefinition}[GHDs of extended subhypergraphs]
	\label{def:ext-ghd}
	Let $H$ be a hypergraph and $H' \cup S_p$ an extended subhypergraph of $H$. 
	A GHD of $H' \cup S_p$ is a tuple $\langle T, (B_u)_{u \in T}, (\lambda_u)_{u \in T} \rangle$, 
	where $T = (N(T),E(T)) $ is a tree, and $B_u$ and $\lambda_u$ are labeling functions, 
	which map to each node $u \in T$ two sets, 
	$B_u \subseteq V(H)$ and $\lambda_u \subseteq E(H) \cup S_p$. 
	For a node $u$, we call $B_u$ the \emph{bag}
	and $\lambda_u$ the \emph{edge cover} of $u$. The set 
	$B(\lambda_u)$ of vertices ``covered'' by $\lambda_u$ is defined as 
	$B(\lambda_u) = \{v \in V(H) \mid v \in f, f \in \lambda_u \}$.
	The functions $\lambda_u$ and $B_u$  have to satisfy the following conditions: 
	\begin{enumerate}
		\item For each node $u \in T$,  either \\
		a) $\lambda_u \subseteq E(H)$ and $B_u \subseteq B(\lambda_u)$, or \\
		b) $\lambda_u = \{s\}$ for some $s \in S_p$ and $B_u  = s$.
		
		\item If, for some $u \in T$,  $\lambda_u = \{s\}$ for some $s \in S_p$, then $u$ is a leaf node.
		
		\item For each $e \in H' \cup S_p$, there is a node $u \in T$ s.t. $e \subseteq B_u$.
		
		\item For each vertex $v \in V(H)$, $\{ u \in T \mid  v \in B_u  \}$ is a connected subtree of $T$. 
	\end{enumerate}
	The \emph{width of a GHD} is defined as $\max\{\lvert \lambda_u\rvert \colon u \in T \}$.
\end{myDefinition}

Clearly, also $H$ itself is an extended subhypergraph of $H$ with $H' = H$ and $S_p = \emptyset$. 
It is readily verified that the above definition of GHD of an extended subhypergraph $H' \cup S_p$ and the definition of 
GHD of a hypergraph $H$ coincide for the special case of taking $H$ as an extended subhypergraph of itself.

In \cite{DBLP:journals/jcss/GottlobLS02}, a normal form of hypertree decompositions was introduced. We will carry 
the notion of normal form over to GHDs of extended subhypergraph. To this end, it is convenient to
first define the set of edges {\em exclusively\/} covered by some subtree of a GHD:

\begin{myDefinition}
	\label{def:excov}
	Let $H' \cup S_p$ be an extended subhypergraph of some hypergraph $H$ and 
	$\mathcal{D} = \langle T, (B_u)_{u \in T}, (\lambda_u)_{u \in T} \rangle$ a GHD for $H' \cup S_p$. 
	For a node $u \in T$, we write $T_u$ to denote the subtree of $T$ rooted at $u$. 
	Moreover, we define the set of edges exclusively covered by $T_{u}$
	as  $\mathit{exCov}(T_{u}) = \{ f \in H'\cup S_p \mid \exists v \in T_{u} : f \subseteq B_v \wedge f \nsubseteq B_u \}$.
\end{myDefinition}

Our normal form of GHDs is then defined as follows:

\begin{myDefinition}[GHD of extended subhypergraphs normal form]
	We say that 
	a GHD $\langle T, (B_u)_{u \in T}, (\lambda_u)_{u \in T} \rangle$ 
	of an extended subhypergraph  $H' \cup S_p$  is in normal form, if for the root node $r$ of $T$, 
	the following property holds: let  $u_1, \dots, u_\ell$ be the child nodes of 
	$r$ in $T$ and let $T_{u_1}, \dots, T_{u_\ell}$ denote the subtrees in $T$ rooted at 
	$u_1, \dots, u_\ell$, respectively. Then 
	$\mathit{exCov}(T_{u_1}), \dots, \mathit{exCov}(T_{u_\ell})$ are precisely the $[V(\lambda_r)]$-components of $H' \cup S_p$.
	Intuitively, each subtree $T_i$ below the root ``covers'' the edges of {\em precisely one\/}
	$[B_r]$-component of $H' \cup S_p$.
\end{myDefinition}

The following lemma is an immediate extensions of previous results for hypergraphs
to extended subhypergraphs.

\begin{myLemma}
	\label{lm:lemma1}
	Let $H' \cup S_p$ be an extended subhypergraph of some hypergraph $H$ and suppose
	that there exists a GHD $\mathcal{D}$ of width $\leq k$ for $H' \cup S_p$. Then there 
	also exists a GHD $\mathcal{D}'$ in normal form of width $\leq k$ for $H' \cup S_p$, such that 
	$B_r$ is a balanced separator of $H' \cup S_p$ for the root node $r$ of~$\mathcal{D}'$. 
\end{myLemma}

\begin{proof}[Proof]
	The lemma combines two results from \cite{DBLP:journals/jcss/GottlobLS02} and \cite{DBLP:journals/ejc/AdlerGG07}, respectively.
	
	Our normal form relaxes the normal form of HDs introduced in Definition 5.1 in \cite{DBLP:journals/jcss/GottlobLS02}. The transformation
	into normal form can be taken over almost literally from the proof of Theorem 5.4 in \cite{DBLP:journals/jcss/GottlobLS02} for establishing the 
	normal form of HDs.
	
	The existence of a balanced separator as the root of a GHD is implicit in the definition of ``hyperlinkedness'' and Theorem 19 in 
	\cite{DBLP:journals/ejc/AdlerGG07}.  Again, it can be easily taken over to our case of an extended subhypergraph (i.e., to 
	take also special edges into account).
	Actually, it can also be easily proved directly by starting off at the root $r$ of an arbitrary GHD 
	of $H' \cup S_p$ in normal form and checking if the components covered by the subtrees below all have size at most 
	$\frac{\lvert H' \cup S_p \rvert}{2}$. If so, we already have the desired form. If not, there must be one subtree $T_{r'}$ 
	rooted at a child $r'$ of $r$, such that 
	$\mathit{exCov}(T_{r'})$
	is greater than $\frac{\lvert H' \cup S_p \rvert}{2}$. We then apply the normal form transformation 
	also to $T_{r'}$ and check recursively if all the components covered by the subtrees below all have size at most 
	$\frac{\lvert H' \cup S_p \rvert}{2} $. By repeating this recursive step, we will eventually reach a node $u$, such that
	$B_u$ is a balanced separator. Then we simply take this node as the root and again apply the normal form transformation
	of the proof of Theorem 5.4 in \cite{DBLP:journals/jcss/GottlobLS02} to this new root node and the subtrees immediately below it.
\end{proof}

\subsubsection{Algorithm Description}
Here we describe Algorithm~\ref{alg:balsep}, which we call \balsep.
For a fixed integer $k \geq 1$, it takes as input a hypergraph $H$ and computes a GHD of $H$ of width $\leq k$ if it exists, or returns \textsc{null} otherwise.
The main of Algorithm~\ref{alg:balsep} consists of a call to the Function~Decompose with parameters $H$ and an empty set of special edges.

The recursive Function~Decompose constitutes the core of the algorithm and, given as input a hypergraph $H'$ and a set of 
special edges $S_p$, computes a GHD of $H' \cup S_p$ of width $\leq k$ if it exists.
Lines~5-12 deal with the two base cases of the algorithm.
If $H' \cup S_p$ has only one edge, we create a decomposition made of a single node $u$ whose label $\lambda_u$ contains the only edge of $H' \cup S_p$ and the bag $B_u = V(H' \cup S_p)$, i.e., all the vertices of the extended subhypergraph, which are also the ones covered by $\lambda_u$.
In a similar way, we deal with the case $\lvert E(H' \cup S_p) \rvert = 2$.
We simply create two nodes $u,v$, one for each edge of $H' \cup S_p$, and then attach $v$ as a child of $u$ and return the decomposition.

If the extended subhypergraph has at least 3 edges, we have to decompose it until we reach one of the two base cases.
In line~13, we initialize the object $\mathit{BalSepIt}$, which is an iterator over the balanced separators of size $\leq k$ of $H' \cup S_p$ with edges in $H$.
The iterator $\mathit{BalSepIt}$ produces, one by one, all the $\ell$-combinations of edges in $H$, for each $\ell \leq k$, to find a balanced separator for $H' \cup S_p$.
Moreover, if all the combinations of full edges fail, the function uses subedges of $H$ to generate separators corresponding to elements of the set $f(H,k)$ of Equation~\ref{eq:bip-subsets}.

In the while loop in lines~14-27, we recursively decompose $H' \cup S_p$.
We are now creating the current node $u$ of the GHD and we have to compute $\lambda_u$, $B_u$ and the children of $u$.
In line~15 we set $\lambda_u$ as the next balanced separator and, in line~16, we fix the bag $B_u = B(\lambda_u)$ as discussed above.
We want to compute a GHD for each $[B_u]$-component of $H' \cup S_p$, in such a way that it will be possible to attach each of them to the current node $u$ without violating any condition of Definition~\ref{def:ext-ghd}.

Function~\ref{alg:compsubh} computes the set of extended subhypergraphs corresponding to $[B_u]$-components of $H' \cup S_p$ and introduces, in each of them, a new special edge $B_u$ for connectedness.
We assume here the existence of a function ConnectedComponents which computes the connected components of a hypergraph in a standard way.
First, Function~\ref{alg:compsubh} computes the hypergraph $H_u = (V_u, E_u)$ resulting from the removal of all vertices in $B_u$ 
from $H' \cup S_p$.
Then, in lines~5-11, it creates a new subhypergraph of $H' \cup S_p$ for each connected component of $H_u$.
The new subhypergraph is stored in a variable $c$, which is a pair consisting of a hypergraph $H$ and a set of special edges $S_p$.
For a single subhypergraph, the new set of special edges $c.S_p$ is composed of the edges of $S_p$ intersecting the current component $\mathit{comp}$ plus a new special edge $s = B_u$ corresponding to the separator $B_u$.
We can then compute the hypergraph $c.H$.
Its set of edges $E$ contains the edges of $H'$ intersecting the current component $\mathit{comp}$ and its set of vertices is the union of $E$ and $c.S_p$.
We then add $c$ to the set of results $\mathit{res}$, which we finally return in line~12.

Back in Function~Decompose, in lines~18-24, we recursively compute a GHD for each extended subhypergraph returned by Function~\ref{alg:compsubh}.
If the decomposition $\mathcal{D}$ returned in line~19 is not \textsc{null}, we add it to the set $\mathit{subDecomps}$ of the children of the current node $u$, otherwise, we set $\mathit{subDecomps}$ to \textsc{null} and break the loop.
At the end of the loop, we check whether $\mathit{subDecomps}$ is \textsc{null}.
If this is the case, it means that one of the recursive calls of Function~Decompose was unsuccessful.
We then have to continue the while loop of lines~14-27 and try the next balanced separator.
In case all the recursive calls of Function~Decompose were successful, Function~\ref{alg:buildghd} builds the resulting GHD and returns it in line~27.
If the algorithm exhausts all the choices of balanced separators for $H' \cup S_p$ and exits the while loop of lines~14-27, it means that it is impossible to create a GHD of width $\leq k$ for $H' \cup S_p$ and we return \textsc{null} in line~28.

Finally, we describe the Function~\ref{alg:buildghd} that, given a bag $B_u$, an edge cover $\lambda_u$, and a set $\mathit{children}$ of GHDs, returns a GHD with root $u$ and children $\mathit{children}$.
We start off creating the node $u$ with labels $B_u$ and $\lambda_u$.
Then, for each child $\mathcal{D} \in \mathit{children}$, we find the node $r$ in $T$ (the tree of $\mathcal{D}$) having $B_r = B_u$ and reroot $T$ to $r$.
Now, for each child $c_r$ of $r$, we attach $c_r$ as a child of $u$.
In other words, we attach every subtree rooted at a child node of $r$ to $u$.
Finally, we return the resulting decomposition.

\begin{algorithm}
	\SetKwInOut{KwPar}{Parameter}
	\SetKwProg{Main}{Main}{}{}
	\SetKwProg{Fn}{Function}{}{}
	\SetKwFunction{Decompose}{Decompose}
	\SetKwFunction{EnumBalSeps}{EnumBalSeps}
	\SetKwFunction{CompSubHyp}{ComputeSubhypergraphs}
	\SetKwFunction{BuildGHD}{BuildGHD}
	\SetKwFunction{Attach}{AttachChild}
	\SetKwFunction{Init}{InitBalSepIterator}
	\SetKwFunction{HasNext}{HasNext}
	\SetKwFunction{Next}{Next}
	\SetKw{Break}{break}
	\SetKw{Continue}{continue}
	
	\DontPrintSemicolon
	
	\KwIn{A hypergraph $H$.}
	\KwPar{An integer $k \geq 1$.}
	\KwOut{A GHD of $H$ of width $\leq k$ if it exists, \textsc{null} otherwise.}
	
	\Main{}{
		Make $H$ globally visible\;
		\Return \Decompose{$H, \emptyset$}\;
	}
	
	\Fn{\Decompose{$H'$: hypergraph, $S_p$: set of special edges}}{
		\If{$\lvert E(H' \cup S_p) \rvert == 1$}{
			\Return node $u$ with $B_u \leftarrow V(H' \cup S_p)$ and $\lambda_u \leftarrow E(H' \cup S_p)$\;
		}
		\If{$\lvert E(H' \cup S_p) \rvert == 2$}{
			Let $e_1, e_2$ be the two edges of $H' \cup S_p$\;
			Create node $u$ with $B_u \leftarrow e_1$ and $\lambda_u \leftarrow \{e_1\}$\;
			Create node $v$ with $B_v \leftarrow e_2$ and $\lambda_v \leftarrow \{e_2\}$\;
			\Attach{$u,v$}\;
			\Return $u$;
		}
		$\mathit{BalSepIt} \leftarrow$ \Init{$H, H', S_p, k$}\;
		\While{\HasNext{$\mathit{BalSepIt}$}}{
			$\lambda_u \leftarrow$ \Next{$\mathit{BalSepIt}$}\;
			$B_u \leftarrow B(\lambda_u)$\;
			$\mathit{subDecomps} \leftarrow \{\}$\;
			\ForEach{$c \in$ \CompSubHyp{$H', S_p, B_u$}}{
				$\mathcal{D} \leftarrow$ \Decompose{$c.H, c.S_p$}\;
				\uIf{$\mathcal{D} \neq$ \textsc{null}}{
					$\mathit{subDecomps} \leftarrow \mathit{subDecomps} \cup \{\mathcal{D}\}$\;
				}\Else{
					$\mathit{subDecomps} \leftarrow$ \textsc{null}\;
					\Break\;
				} 
			}
			\If{$\mathit{subDecomps} ==$ \textsc{null}}{
				\Continue\;
			}
			\Return \BuildGHD{$B_u, \lambda_u, \mathit{subDecomps}$}\;
		}
		\Return \textsc{null}\;
	}
	
	\caption{\balsep{}}
	\label{alg:balsep}
\end{algorithm}

\begin{function}
	\SetKwFunction{Reroot}{Reroot}
	\SetKwFunction{Children}{Children}
	\SetKwFunction{Attach}{AttachChild}
	
	\DontPrintSemicolon
	\KwIn{A set of vertices $B_u$, a set of $\lambda_u$, a set of GHDs $\mathit{children}$.}
	\KwOut{A GHD with root $u$ and children $\mathit{children}$.}
	
	\Begin{
		Create node $u$ with $B_u$ and $\lambda_u$\;
		\ForEach{$\mathcal{D} \in \mathit{children}$}{
			Let $T$ be the tree structure of $\mathcal{D}$\;
			$\hat{r} \leftarrow$ \Reroot{$T, B_u$}\;
			\ForEach{$c_{\hat{r}} \in $ \Children{$\hat{r}$}}{
				\Attach{$u, c_{\hat{r}}$}
			}
		}
		\Return $u$\;
	}
	
	\caption{BuildGHD($B_u, \lambda_u, \mathit{children}$)}
	\label{alg:buildghd}
\end{function}

\begin{function}
	\SetKwFunction{ConnComp}{ConnectedComponents}
	
	\DontPrintSemicolon
	\KwIn{A hypergraph $H'$, a set of special edges $S_p$, a set of vertices $B_u$.}
	\KwOut{The set of subhypergraphs of $H' \cup S_p$ w.r.t. $B_u$.}
	
	\Begin{
		$V_u \leftarrow V(H') \setminus B_u$\;
		$E_u \leftarrow \{ e \cap V_u \mid e \in E(H' \cup S_p) \}$\;
		$\mathit{res} \leftarrow \{\}$\;
		\ForEach{$\mathit{comp} \in$ \ConnComp{$V_u, E_u$}}{
			$c \leftarrow$ initialize pair $(H=\textsc{null}, S_p=\textsc{null})$\;
			$c.S_p \leftarrow \{ s \in S_p \mid s \cap \mathit{comp} \neq \emptyset \} \cup \{B_u\}$\;
			$E \leftarrow \{ e \mid e \in E(H') \land e \cap \mathit{comp} \neq \emptyset \}$\;
			$V  \leftarrow V(E) \cup V(c.S_p)$\;
			$c.H \leftarrow (V, E)$\;
			$\mathit{res} \leftarrow \mathit{res} \cup \{ c \}$\;
		}
		\Return $\mathit{res}$\;
	}
	
	\caption{ComputeSubhypergraphs($H', S_p, B_u$)}
	\label{alg:compsubh}
\end{function}

\subsubsection{Soundness and Completeness}
Here we prove that Algorithm~\ref{alg:balsep} is sound and complete.
\begin{myTheorem}
	Let $H$ be a hypergraph and $k \geq 1$ an integer. Algorithm~\ref{alg:balsep} called on $H$ with parameter $k$ returns a GHD of $H$ of width $\leq k$ if and only if $\ghw(H) \leq k$.
\end{myTheorem}

We prove the soundness and the completeness of Algorithm~\ref{alg:balsep} separately.
Nevertheless, we want to point out that the main procedure of the algorithm consists of a call to Function~Decompose with input $H$ and an empty set of special edges.
Thus, in the next proofs, we will actually prove that Function~Decompose called on $(H',S_p)$ with parameter $k$ returns a GHD of $H' \cup S_p$ of width $\leq k$ if and only if $\ghw(H' \cup S_p) \leq k$ w.r.t. Definition~\ref{def:ext-ghd}.
Note that in case of a hypergraph $H$ and a set of special edges $S_p = \emptyset$, Definition~\ref{def:ext-ghd} coincides with the usual definition of GHD.

\begin{proof}(Soundness)
	We show that if Function~Decompose called on $(H',S_p)$ with parameter $k$ returns a GHD of $H' \cup S_p$ of width $\leq k$, then such a decomposition actually exists and $\ghw(H' \cup S_p) \leq k$.
	We proceed by induction over the size of $H' \cup S_p$, i.e., $\lvert E(H' \cup S_p) \rvert$.
	For the base case, we assume $\lvert H' \cup S_p \rvert \leq 2$.
	In case $\lvert H' \cup S_p \rvert = 1$, we return a GHD made of a single node whose $\lambda$-label consists of the only edge $H' \cup S_p$, which also cover all of the vertices of the hypergraph.
	Such a decomposition has width 1 and clearly satisfies all the conditions of Definition~\ref{def:ext-ghd}.
	In case $\lvert H' \cup S_p \rvert = 2$, we create two nodes $u,v$, each one corresponding to an edge of $H' \cup S_p$ and we attach $v$ as a child of $u$.
	Note that both $u$ and $v$ are leaf nodes. It is easy to verify that also in this case we return a valid GHD of $H' \cup S_p$.
	
	For the induction step, suppose that the recursive function Decompose correctly returns a GHD of $H' \cup S_p$ of
	width $\leq k$ for each $H' \cup S_p$ such that $\lvert H' \cup S_p \rvert \leq j$, for some $j \geq 2$.
	Now suppose that  $\lvert H' \cup S_p \rvert = j + 1 \geq 3$ and that Function~Decompose($H',S_p$) returns a GHD of $H' \cup S_p$ of width $\leq k$.
	We have to show that then there indeed exists such a GHD.
	
	Algorithm~\ref{alg:balsep} only returns a GHD in line~27.
	The program successfully reaches this line only if the following happens:
	\begin{itemize}
		\item In line~15, a balanced separator $\lambda_u$ of $H' \cup S_p$ is chosen.
		\item In line~18, the extended subhypergraphs of $H' \cup S_p$ w.r.t. $B_u = B(\lambda_u)$ are computed;
		in particular, each extended subhypergraph corresponds to a $[B_u]$-component of $H' \cup S_p$ plus a new special edge $s = B_u$.
		\item For each extended subhypergraph $c.H \cup c.S_p$, the call to Function~Decompose in line~19 is successful, i.e., it returns a GHD of $c.H \cup c.S_p$ of width $\leq k$.
	\end{itemize}
	
	We are assuming $\lvert H' \cup S_p \rvert \geq 3$ and $\lambda_u$ is a balanced separator of $H' \cup S_p$.
	Let $C_1, \dots, C_\ell$ be the $\ell$ $[B_u]$-components of $H' \cup S_p$.
	All extended subhypergraphs $C_i \cup \{B_u\}$ are strictly smaller than $j$.
	Hence, by the induction hypothesis, for each $i \in \{1,\dots,\ell\}$, there indeed exists a GHD $\mathcal{D}_i = \langle T_i, (B_{i,u})_{u \in T}, (\lambda_{i,u})_{u \in T} \rangle$ of width $\leq k$ for the extended subhypergraph $C_i \cup \{B_u\}$.
	
	It is left to show that Function~\ref{alg:buildghd} correctly constructs a GHD of width $\leq k$ of $H' \cup S_p$.
	For each $\mathcal{D}_i$, let $\hat{r}_i$ be the node of $T_i$ with $B_{i,\hat{r}_i} = B_u$ and $\lambda_{i,\hat{r}_i} = \{B_u\}$.
	By construction the node $\hat{r}_i$ exists in every $\mathcal{D}_i$, it is always a leaf and it has the same bag and $\lambda$-label everywhere.
	Let $T_i = (N(T_i), E(T_i))$ be the tree structure of $\mathcal{D}_i$ and, w.l.o.g., assume that the node sets $N(T_i)$ are pairwise disjoint.
	We define the tree structure $T = (N(T),E(T))$ and the functions $B_u$ and $\lambda_u$ of $\mathcal{D}$ as follows:
	\begin{itemize}
		\item $N =  (N(T_1) \setminus \{ \hat{r}_1\} ) \cup \dots \cup  (N(T_\ell) \setminus \{ \hat{r}_\ell\} ) \cup \{r\}$,
		where $r$ is a new (root) node. 
		\item For the definition of $E(T)$ recall that each $\hat{r}_i$ is a leaf node in its decomposition. Let $e_i$ denote
		the edge between $\hat{r}_i$ and its parent. Then we define $E(T)$ as
		$E(T) = (E(T_1) \setminus \{ e_1\} ) \cup \dots \cup  (E(T_\ell) \setminus \{ e_\ell\} ) \cup R$ with
		$R = \{ [r,\hat{r}_1], \dots, [r,\hat{r}_\ell]\}$.
		\item $\lambda_r = \lambda_u$ and $B_r = B_u$.
		\item For every $v \in N \setminus \{r\}$, there exists exactly one $i$, such that $v \in N(T_i)$. 
		We set $\lambda_v = \lambda_{i,v}$ and $B_v = B_{i,v}$. 
	\end{itemize}
	Intuitively, the GHD $\mathcal{D}$ is obtained by taking in each GHD $\mathcal{D}_i$ the node $\hat{r}_i$ as the root
	node and combining all GHDs $\mathcal{D}_i$  to a single GHD by merging their root nodes to a single node $r$. 
	This is possible since all nodes $\hat{r}_i$ have the same $\lambda$-labels and bags. 
	It is easy to verify that the resulting GHD is indeed a GHD of width $\leq k$ of the extended subhypergraph $H' \cup Sp$.
\end{proof}

\begin{proof}(Completeness)
	For any hypergraph $H'$ and set of special edges $S_p$, we prove that if $\ghw(H' \cup S_p) \leq k$, then Function~Decompose on input $(H',S_p)$ returns a GHD of width $\leq k$ of $H' \cup S_p$.
	Again we proceed by induction on $\lvert H' \cup S_p \rvert$.
	
	The base case of $\lvert H' \cup S_p \rvert \leq 2$ is dealt with in lines~5-12 of Algorithm~\ref{alg:balsep}.
	In this case, we simply construct a GHD of width 1 for $\lvert H' \cup S_p \rvert$ and return it.
	
	For the induction step, suppose $\lvert H' \cup S_p \rvert \leq j$, for some $j \geq 2$, and $\ghw(H' \cup S_p) \leq k$.
	Then, Function~Decompose on input $(H',S_p)$ returns a GHD of width $\leq k$ of $H' \cup S_p$.
	Now assume that $\lvert H' \cup S_p \rvert = j + 1 \geq 3$ and that $\ghw(H' \cup S_p) \leq k$.
	We have to show that Function~Decompose on input $(H',S_p)$ returns a GHD of width $\leq k$ of $H' \cup S_p$.
	
	By Lemma~\ref{lm:lemma1}, we may assume w.l.o.g., that GHD $\mathcal{D} = \langle T, (B_u)_{u \in T}, (\lambda_u)_{u \in T} \rangle$ 
	is in normal form and that $B_r$ is a balanced separator of $H' \cup S_p$ for the root node $r$ of~$\mathcal{D}$. 
	Let $S \subseteq E(H)$ denote the $\lambda$-label of $r$, i.e., $\lambda_r = S$.
	
	When Function~Decompose is called on input $(H',S_p)$, the while loop in lines~14-27 eventually generates all possible balanced separators of size $\leq k$ of $H' \cup S_p$, unless it returns on line~27 before the end of the loop.
	Remember that the object $\mathit{BalSepIt}$ not only generates edge separators with edges in $E(H)$, i.e., the original hypergraph on which Algorithm~\ref{alg:balsep} is called, 
	but it also uses edges in $f(H,k)$, i.e., subedges of edges in $E(H)$.
	Thus, at some point, in line~15, we will choose the separator $\lambda_u = S$, equivalently, $B_u = B_r$.
	
	Let $C_1, \dots, C_\ell$ denote the $[B_u]$-components of $H' \cup S$.
	Since $\mathcal{D}$ is in normal form, we know that the root node $r$ has $\ell$ child nodes such that 
	$C_i = \mathit{exCov}(T_{i})$, where $T_{i}$ is the subtree in $T$ rooted at $n_i$ for $i \in \{1, \dots, \ell\}$. 
	Recall from Definition~\ref{def:excov} that we write $\mathit{exCov}(T_{i})$ to denote the 
	set of edges exclusively covered by $T_{i}$.
	
	Now consider the extended subhypergraph $C_i \cup \{B_u\}$ for arbitrary $i \in \{1, \dots, \ell\}$.
	Since  $B_r$ is a balanced separator of $H' \cup S_p$, we have $\lvert C_i \cup \{B_u\} \rvert \leq j$.
	Moreover, there exists a GHD $\mathcal{D}_i$ of $C_i \cup \{B_u\}$, namely the subtree of $\mathcal{D}$ induced by the nodes in $T_i$ plus $r$. 
	Hence, by the induction hypothesis, calling Function~Decompose with the input corresponding to the 
	extended subhypergraph $C_i \cup \{B_u\}$ returns a valid GHD. 
	This means that, in our call of Function~Decompose with input $(H',Sp)$, we have the following behavior:
	\begin{itemize}
		\item On line~19, Function~Decompose is called recursively for all extended subhypergraphs $C_i \cup \{B_u\}$.
		\item Each call of the Function~Decompose returns a GHD for the respective extended subhypergraph.
		\item The results of these recursive calls are collected in line~21 in the variable $\mathit{subDecomps}$.
		\item Hence, after exiting the loop in lines~18-26, the return statement in line~27 is executed.
	\end{itemize}
	The call to Function~\ref{alg:buildghd} correctly produces the desired GHD as discussed in the soundness proof. 
	Finally, Function~Decompose indeed returns a GHD of width $\leq k$ of $H' \cup S_p$.
\end{proof}

\section{HyperBench Benchmark and Tool}
\label{sec:hyperbench}

In this section, we describe \emph{HyperBench} - our new benchmark and web tool.
We first introduce our system and test environment used for the experiments, 
then we present a new method to extract simple conjunctive queries from complex SQL queries.
Finally, we describe the CQs and CSPs we have collected.

\subsection{System and Test Environment}
Our system is composed of two libraries: a C++ library with implementations of the algorithms and a Java library for processing of SQL queries and hypergraphs.
We have extended the algorithm \detkdecomp from~\cite{DBLP:journals/jea/GottlobS08} and based our program on a new implementation, which we call \newdetkdecomp.
The resulting C++ library contains the implementations of all the algorithms presented in Sections~\ref{sec:ghd-algs}~and~\ref{sec:fhd-algs} and comprises around 8500 lines of code.
We designed our Java library \emph{hg-tools} for preprocessing of SQL queries and collecting hypergraph statistics.
It uses the open source libraries \emph{JSqlParser}~\cite{sw:jsqlp} for SQL processing and \emph{JGraphT}~\cite{jgrapht} to deal with graph data structures.
Our two libraries are available at~\url{https://github.com/dmlongo/newdetkdecomp} and \url{https://github.com/dmlongo/hgtools}.

All the experiments reported in this paper were performed on a cluster of 10 workstations each running Ubuntu 16.04.
Every workstation has the same specification and is equipped with two Intel Xeon E5-2650 (v4) processors each having 12 cores and 256-GB main memory.
Since all algorithms are single-threaded, we could run several experiments in parallel.
For all upcoming runs of our algorithms we set a timeout of 3600s.

\subsection{Translation of complex SQL Queries into Hypergraphs}
\label{sec:translation}
Since the applicability of structural decomposition methods is limited to conjunctive queries, we have devised a strategy to transform also more complex queries into (a collection of) hypergraphs.
Our program \emph{hg-tools} takes as input a complex SQL query and produces a collection of simpler SQL queries.
When facing a complex query, e.g., one containing nested queries, the idea is to identify simple queries, extract them and transform them separately.

A conjunctive query can be expressed in SQL as a SELECT-FROM-WHERE statement in which the WHERE clause is only allowed to be a conjunction of equality conditions.
Moreover, such a query must not contain negation, disjunction and subqueries, i.e., nested SELECT statements.
For our purposes, we neglect the SELECT clause because in our experiments we focus on computing decompositions instead of answering the queries.
Hence, only the hypergraph structure determined by the FROM and WHERE clauses is important.

Query~\ref{q:simple-query} is an example of an SQL query.
The FROM clause uses two instances of the relation $\mathit{tab(a,b,c)}$ and in the WHERE clause there are some conditions on the two relation instances.
The condition on line~3 is a join and falls in the scope of conjunctive queries.
Though the condition on line~4 is not conjunctive, it is just a comparison with a constant value and it does not influence the query structure.
Nevertheless, the condition on line~5 involves a negation, thus Query~\ref{q:simple-query} is not a conjunctive query.
This kind of queries is not allowed by our framework.

\begin{lstlisting}[caption={A simple SQL query},label={q:simple-query}]
SELECT *
FROM tab t1, tab t2
WHERE t1.a = t2.a
  AND t1.b > 5
  AND t1.c <> t2.c;
\end{lstlisting}

Nevertheless, what really matters in defining the query structure are the relationships between the different variables involved in the query, in particular join conditions.
For this reason, when faced with a query that is not conjunctive, we consider a simplified version that contains the conjunctive core of the original query.
In case of Query~\ref{q:simple-query}, we drop lines~4-5 and perform our experiments on the rest of the query.

Real world SQL queries can be rather complicated and, in particular, they can contain nested SELECT statements, which we simply call \emph{subqueries}.
The presence of subqueries automatically makes the query non conjunctive, but we have decided to extract the single queries and analyze their conjunctive cores separately.
Query~\ref{q:nested-queries} is a query with nested SELECT statements.
The subquery on lines~4-6 can be examined separately.
The subquery on lines~7-9 cannot be examined separately because it contains a reference to a table defined in an external query and, more precisely, it requires the evaluation of the subquery for the value of the current row of that table.

\begin{lstlisting}[caption={A complex SQL query},label={q:nested-queries}]
SELECT *
FROM tab t1, tab t2
WHERE t1.a = t2.a
  AND t1.b IN (SELECT tab.b
               FROM tab
               WHERE tab.c == 'ok')
  AND EXISTS (SELECT *
              FROM differentTable dt
              WHERE dt.a = t1.a);
\end{lstlisting}

In SQL, a subquery can appear in different places in a query.
Depending in which statement or condition it appears, we treat it differently:
\begin{itemize}
	\item If a query is of the form $q_1 \circ \dots \circ q_n$, where each $q_i$ is a query and $\circ \in \{\cup, \cap, \setminus\}$, we extract the single queries $q_i$ and process them separately.
	\item If a subquery appears in the FROM clause, we convert it into a view (see Section~\ref{sec:translation-hg}).
	\item If a subquery contains a reference to an external query, as in lines~7-9 of Query~\ref{q:nested-queries}, it must be discarded.
\end{itemize}

\subsection{Extracting Simple Queries}
\label{sec:extract-simple-queries}
In general, an SQL query can contain subqueries.
We want to extract them and either integrate them into the \emph{main} query, when possible, or analyze them separately.
In order to do that in such a way that the end result closely resembles the original query as much as possible,
we build a graph representing the dependencies between the subqueries.
At the end of the process, we extract queries which are independent and eliminate those which are mutually dependent.

We say that a subquery $s_1$ \emph{depends} on a subquery $s_2$ if the result of $q_1$ can be computed only after computing the result of $q_2$.
The dependency graph of a query $Q$ is a graph $G=(S,D)$ where the set $S$ contains the nodes corresponding to subqueries of $Q$
and $(s_1,s_2) \in D$ is an arrow, for $s_1,s_2 \in S$, if $s_1$ depends on $s_2$.
Given a query $Q$, we create its dependency graph $G$ as follows:
\begin{enumerate}
	\item Create a node $q \in S$ representing the outer query.
	\item For each nested query $s_i$ of $q$, create a new node $s_i \in S$
	and an edge $(q,s_i) \in D$.
	\item If $s_i$ contains a reference to a table defined in any ancestor $s_j$,
	create an edge $(s_i,s_j)$.
	\item Recursively examine $s_i$ for nested queries.
\end{enumerate}
Once we have built the graph, we identify the nodes which are involved in cycles and eliminate them.
In particular, we consider $q$ as a root and navigate the graph.
Whenever we find a node having an edge pointing at an ancestor, we eliminate it together with all of its incoming and outgoing edges.
Eventually, we end up with a forest in which we extract a query from each node.

\begin{figure}
	\centering
	\begin{tikzpicture}
	\node[circle,draw,minimum size=1mm]	(q)		at	(1.5,2)	{$q$};
	\node[circle,draw,minimum size=1mm]	(s1)	at	(0,0)	{$s_1$};
	\node[circle,draw,minimum size=1mm]	(s2)	at	(3,0)	{$s_2$};
	
	\draw	[->]	(q) to (s1);
	\draw	[->]	(q) to (s2);
	\draw	[->]	(s2) to [bend left=45] (q);
	\end{tikzpicture}
	\caption{Dependency graph of Query~\ref{q:nested-queries}}
	\label{fig:dep-graph}
\end{figure}
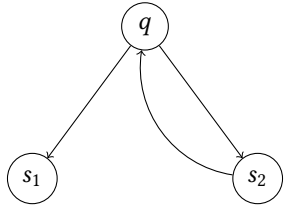

The dependency graph $G$ of Query~\ref{q:nested-queries} is shown in Figure~\ref{fig:dep-graph}.
The node $q$ corresponds to the outer query and it is the root of the graph.
The nodes $s_1$ and $s_2$ represent the subqueries in lines~4-6 and lines~7-9 of Query~\ref{q:nested-queries}, respectively.
As the result of Query~\ref{q:nested-queries} can be computed only after computing the subqueries,
the two edges $(q,s_1)$ and $(q,s_2)$ are present in the graph.
Since $s_2$ refers to the table $t_1$ defined in $q$ (line~9), $G$ contains also the edge $(s_2,q)$.
After the creation, we look for cycles in the graph.
In this case, we see that there is no way to evaluate $s_2$ independently from $q$.
Then, we remove $s_2$ and all of its incident edges.
Finally, we extract a simple query from each node of the remaining graph.

\subsection{Converting Simple Queries into Hypergraphs}
\label{sec:translation-hg}
Once we have extracted and simplified subqueries, we are left with simple SQL queries of the type
\begin{equation}
\label{eq:simple-sql}
\mathrm{SELECT} \; r_{i_1}.A_{j_1}, \dots, r_{i_z}.A_{j_z} \; \mathrm{FROM} \; r_1 ,\dots, r_m \; \mathrm{WHERE} \; \mathit{cond}
\end{equation}
such that $\mathit{cond}$ is a conjunction of conditions of the form $ r_i.A = r_j.B $ or $ r_i.A = c $, where $c$ is a constant.
Such queries are equivalent to conjunctive queries, thus it is easy to draw a connection to a CQ and transform it into a hypergraph.
Nevertheless, in our case it makes more sense to go directly from an SQL query to the hypergraph.

Let $Q$ be an SQL query of the form~(\ref{eq:simple-sql}), then the hypergraph $H_Q = (V(H_Q),E(H_Q))$ corresponding to $Q$ is obtained as follows.
We first build the hypergraph induced by the FROM clause.
Consider a relation $r(A_1,\dots,A_\ell)$ in the FROM clause.
For each attribute $A_i$ of $r$, we create a vertex $v_{A_i} \in V(H_Q)$.
Then, we create the edge $r = \{v_{A_1},\dots,v_{A_\ell}\} \in E(H_Q)$.
Now, we modify the hypergraph according to the conditions in the WHERE clause.
Let $\mathit{cond}$ be such a condition. It can be of two forms:
\begin{itemize}
	\item If $\mathit{cond}$ is of the form $r_i.A = r_j.B$, we merge vertices $v_A$ and $v_b$ and modify their incident edges.
	W.l.o.g. assume $v_A$ itself becomes the merged vertex. For each edge $r \in \{ e \in E(H_Q) \mid v_B \in e \}$, we remove $r$ from $E(H_Q)$ and add a new edge $r' = (r \setminus \{v_B\}) \cup \{v_A\}$.
	\item If $\mathit{cond}$ is of the form $r_i.A = c$, with $c$ constant, we remove $v_A$ from $V(H_Q)$ and, for each edge $r \in \{ e \in E(H_Q) \mid v_A \in e \}$, we remove $r$ from $E(H_Q)$ and add a new edge $r' = r \setminus \{v_A\}$.
\end{itemize}
At the end of this procedure, we eliminate empty edges and multiple edges.
Also, in our setting SELECT clauses do not contribute to the query structure, thus we simply ignore them.

We also take into account logical views.
SQL views are virtual tables which are recreated every time the view is called, thus we have decided to expand the main query by adding the view inside it.
Consider Query~\ref{q:convert-with-view}.
The view \emph{crossView} has a cross structure composed by the two relations $t_1$,$t_2$ intersecting in the node $b$ in the middle, see Figure~\ref{fig:convert-hg}~(a).
The main query adds two relations $t_1$,$t_2$ which are distinct from the ones in the view \emph{crossView} and intersect it in four points, thus creating two cycles.
The end result is depicted in Figure~\ref{fig:convert-hg}~(b).
The resulting query can finally be converted into a hypergraph with the algorithm described above.

\begin{lstlisting}[caption={An SQL query with a view},label={q:convert-with-view}]
WITH crossView AS (
    SELECT t1.a a1, t1.c c1, t2.a a2, t2.c c2
    FROM tab t1, tab t2
    WHERE t1.b = t2.b
)
SELECT *
FROM tab t1, tab t2, crossView cr
WHERE t1.a = cr.a1
  AND t1.c = cr.a2
  AND t2.a = cr.c1
  AND t2.c = cr.c2;
\end{lstlisting}

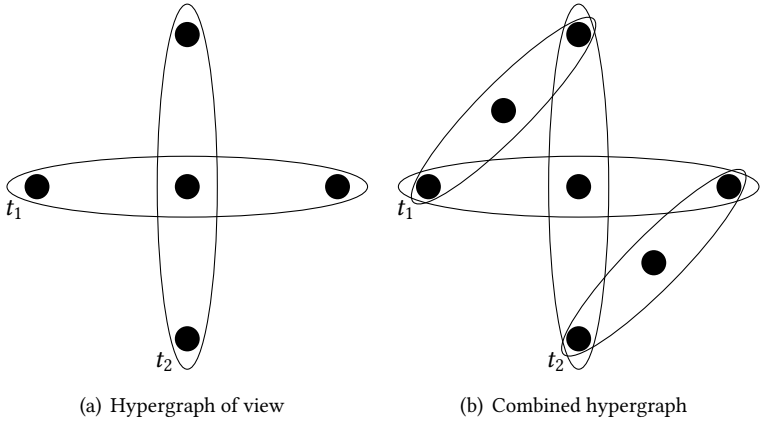
\begin{figure}
	\centering
	\subfigure[Hypergraph of view]{
		\begin{tikzpicture}
		\node[circle,fill,minimum size=1mm]	(t1a)	at	(0,2)	{};
		\node[circle,fill,minimum size=1mm]	(t1b)	at	(2,2)	{};
		\node[circle,fill,minimum size=1mm]	(t1c)	at	(4,2)	{};
		\node[circle,fill,minimum size=1mm]	(t2a)	at	(2,0)	{};
		\node[circle,fill,minimum size=1mm]	(t2c)	at	(2,4)	{};
		
		\node	(t1)	at	(-0.3,1.7)	{$t_1$};
		\node	(t2)	at	(1.7,-0.3)	{$t_2$};
		
		\draw (2,2) ellipse [x radius=2.4, y radius=0.4];
		\draw[rotate around={90:(2,2)}] (2,2) ellipse [x radius=2.4, y radius=0.4];
		\end{tikzpicture}
	}
	\subfigure[Combined hypergraph]{
		\begin{tikzpicture}
		\node[circle,fill,minimum size=1mm]	(t1a)	at	(0,2)	{};
		\node[circle,fill,minimum size=1mm]	(t1b)	at	(2,2)	{};
		\node[circle,fill,minimum size=1mm]	(t1c)	at	(4,2)	{};
		\node[circle,fill,minimum size=1mm]	(t2a)	at	(2,0)	{};
		\node[circle,fill,minimum size=1mm]	(t2c)	at	(2,4)	{};
		\node[circle,fill,minimum size=1mm]	(o1)	at	(1,3)	{};
		\node[circle,fill,minimum size=1mm]	(o2)	at	(3,1)	{};
		
		\node	(t1)	at	(-0.3,1.7)	{$t_1$};
		\node	(t2)	at	(1.7,-0.3)	{$t_2$};
		
		\draw (2,2) ellipse [x radius=2.4, y radius=0.4];
		\draw[rotate around={90:(2,2)}] (2,2) ellipse [x radius=2.4, y radius=0.4];
		\draw[rotate around={45:(1,3)}] (1,3) ellipse [x radius=1.7, y radius=0.35];
		\draw[rotate around={45:(3,1)}] (3,1) ellipse [x radius=1.7, y radius=0.35];
		\end{tikzpicture}
	}
	\caption{Hypergraphs for Query~\ref{q:convert-with-view}}
	\label{fig:convert-hg}
\end{figure}

\subsection{From CSPs to Hypergraphs}
\label{sec:cspHypergraphs}
An important part of our benchmark consists of instances of Constraint Satisfaction Problems.
The set we have collected presents different characteristics w.r.t. the ones found in CQs,
thus their analysis offers a more varied picture of the hypergraphs encountered in applications.
For this reason we have retrieved CSP instances which have a significant practical aspect.

The source of most of our CSPs is the website XCSP~\cite{xcsp}.
XCSP3 is an XML-based format used to represent constraint satisfaction problems.
The language offers a wide variety of options to represent the most common constraints and frameworks,
making it a solid intermediate format between different solvers.
They also organize solver competitions for which they use their instances as a benchmark.

From XCSP we have selected a total of 1,953 instances with less than 100 extensional constraints
such that all constraints are extensional.
These instances are divided into CSPs
from concrete applications, called \textit{CSP Application} in the sequel (\cspApp{} instances), 
and randomly generated CSPs, called \textit{CSP Random} below (\cspRand{} instances).

The instances we have fetched from the website are written in well-structured XML files
in which variables and constraints are explicitly defined through the use of specific XML tags.
The transformation of these instances into hypergraphs did not require a specific methodology
since the authors of the XCSP3 format provide an extensive library for parsing the instances
where most of the process is already automatized.
Obviously, we still had to convert the object in memory into a hypergraph.
To this end, we have reimplemented the behavior of some callback methods in such a way that,
whenever the program reads a variable, it adds a vertex to the hypergraph, and,
whenever it reads a constraint, it adds an edge containing the
vertices corresponding to the variables affected by the constraint.

Our collection of CSPs also includes a third class, which we call \emph{CSP Other}.
These instances have been used in previous hypertree width analyses available at~\url{https://www.dbai.tuwien.ac.at/proj/hypertree/}.
This set contains interesting examples coming from industry and a variety of different test examples~\cite{ganzow2005csp}.
In particular, a part of the hypergraphs is obtained from Daimler Chrysler and represents circuits and systems.
A second part is a hypergraph translation of the circuits belonging to the well-known benchmark library
of the IEEE International Symposium on Circuits and Systems (ISCAS).
Finally, some hypergraphs correspond to grids extracted from pebbling problems.
Since the instances are provided already as hypergraphs, no additional processing was necessary to incorporate them.

\subsection{Hypergraph Benchmark}
Our benchmark contains \hgTotal{} hypergraphs, which have been converted from CQs and CSPs collected from various sources.
Out of these \hgTotal{} hypergraphs, \hgNew{} hypergraphs have never been used in a hypertree width analysis before.
The hypertree width of \cqOld{} CQs and of \cspOld{} CSPs has been analyzed in~\cite{DBLP:journals/jea/GottlobS08}, \cite{berg2017maxsat}, and/or \cite{DBLP:journals/pvldb/BonifatiMT17,DBLP:conf/www/BonifatiMT19}.
An overview of all instances of CQs and CSPs is given in Table~\ref{tab:bench-overview}. 
They have been collected from various publicly available benchmarks and repositories
of CQs and CSPs. In the first column, the names of each collection of CQs and CSPs are given together with 
references where they were first published. In the second column we display the number of hypergraphs
extracted from each collection. The $\hw$ of the CQs and CSPs in our benchmark will be discussed
in detail in Section~\ref{sec:hg-stats}. To get a first feeling of the $\hw$ of the various sources, 
we mention the number of cyclic hypergraphs (i.e., those with $\hw\geq 2$) in the last column.
When gathering the CQs, we proceeded as follows: of the huge benchmark reported in~\cite{DBLP:journals/pvldb/BonifatiMT17}, 
we have only included CQs, which were detected as having $\hw \geq 2$ in \cite{DBLP:journals/pvldb/BonifatiMT17}.
Of the other huge repository reported in~\cite{DBLP:conf/www/BonifatiMT19}, we included the hypergraphs corresponding to the 273,974 unique SPARQL queries with $\hw \geq 2$. Even though the queries are unique, most of them share the same hypergraph structure. Thus, after removing duplicates on the hypergraph level, we ended up with 354 unique hypergraphs with $\hw \geq 2$.
Of the big repository reported in~\cite{DBLP:conf/sigmod/JainMHHL16}, we have included those CQs, which are not trivially acyclic (i.e., they have at least 3 atoms). Of all the small collections of queries, we have included all. It follows a detailed description of the different benchmarks.

\begin{table}
	\centering
	\caption{Overview of benchmark instances}
	\begin{tabular}{llrr}
		\toprule
		& \multicolumn{1}{l}{Benchmark}		& \multicolumn{1}{c}{No. instances}		& \multicolumn{1}{c}{$\hw\geq 2$} \\
		\midrule
		\parbox[t]{2mm}{\multirow{11}{*}{\rotatebox[origin=c]{90}{CQs}}}
		& \textsc{SPARQL}~\cite{DBLP:journals/pvldb/BonifatiMT17}									& 70 (out of 26,157,880)	& 70	\\
		& \textsc{Wikidata}~\cite{DBLP:conf/www/BonifatiMT19}										& 354 (out of 273,947)		& 354	\\
		& \textsc{LUBM}~\cite{DBLP:conf/pods/BenediktKMMPST17,DBLP:journals/ws/GuoPH05}				& 14						& 2		\\
		& \textsc{iBench}~\cite{DBLP:conf/pods/BenediktKMMPST17,DBLP:journals/pvldb/ArocenaGCM15}	& 40						& 0		\\
		& \textsc{Doctors}~\cite{DBLP:conf/pods/BenediktKMMPST17,DBLP:conf/icde/GeertsMPS14}		& 14						& 0		\\
		& \textsc{Deep}~\cite{DBLP:conf/pods/BenediktKMMPST17}										& 41						& 0		\\
		& \textsc{JOB} (IMDB)~\cite{DBLP:journals/vldb/LeisRGMBKN18}								& 33						& 7		\\
		& \textsc{TPC-H}~\cite{BenediktCQs,tpch}													& 29						& 1		\\
		& \textsc{TPC-DS}~\cite{tpcds}																& 228						& 5		\\
		& \textsc{SQLShare}~\cite{DBLP:conf/sigmod/JainMHHL16}										& 290 (out of 15,170)		& 1		\\
		& \textsc{Random}~\cite{DBLP:journals/vldb/PottingerH01}									& 500						& 464	\\
		\midrule
		\parbox[t]{2mm}{\multirow{3}{*}{\rotatebox[origin=c]{90}{CSPs}}}
		& \textsc{Application}~\cite{xcsp} 															& 1,090  					& 1,090 \\
		& \textsc{Random}~\cite{xcsp}																& 863   					& 863	\\
		& \textsc{Other}~\cite{DBLP:journals/jea/GottlobS08,berg2017maxsat}							& 82						& 82	\\ 
		\midrule
		\multicolumn{1}{r}{} & \multicolumn{1}{l}{\textit{Total:}} 									& \textit{\hgTotal{}}		& \textit{2,939} \\
		\bottomrule
	\end{tabular}%
	\label{tab:bench-overview}%
\end{table}%

Our benchmark contains \cqApp{} CQs from five main sources~\cite{BenediktCQs,DBLP:conf/pods/BenediktKMMPST17,DBLP:journals/pvldb/BonifatiMT17,DBLP:conf/www/BonifatiMT19,DBLP:conf/sigmod/JainMHHL16} and a set of \cqRand{} randomly generated queries using the query generator of~\cite{DBLP:journals/vldb/PottingerH01}. 
In the sequel, we shall refer to the former queries as {\it CQ Application\/}, and to the latter
as {\it CQ Random\/}.
The CQs analyzed in~\cite{DBLP:journals/pvldb/BonifatiMT17} constitute a big repository of CQs -- namely 26,157,880 CQs stemming from
SPARQL queries. The queries come from real-users of SPARQL endpoints and their hypertree width was already determined in~\cite{DBLP:journals/pvldb/BonifatiMT17}. Almost all of these CQs were shown to be acyclic. Our analysis comprises 70 CQs from~\cite{DBLP:journals/pvldb/BonifatiMT17}, 
which (apart from few exceptions) are essentially the ones in~\cite{DBLP:journals/pvldb/BonifatiMT17} with $\hw \geq 2$. 
In particular, we have analyzed all 8 CQs with highest $\hw$ among the CQs analyzed in~\cite{DBLP:journals/pvldb/BonifatiMT17} 
(namely, $\hw = 3$).
Bonifati et al. carried on this line of work and examined a bigger repository of SPARQL queries coming from Wikidata in~\cite{DBLP:conf/www/BonifatiMT19}.
This repository of 208,215,209 SPARQL queries was originally released by Malyshev et al. with the study in~\cite{DBLP:conf/semweb/MalyshevKGGB18}.
Bonifati et al. kindly sent us the unique 273,947 SPARQL queries with $\hw \geq 2$ examined in~\cite{DBLP:conf/semweb/MalyshevKGGB18}.
We extracted 354 different hypergraphs and all of them have $\hw = 2$.

The \textsc{LUBM}~\cite{DBLP:journals/ws/GuoPH05}, \textsc{iBench}~\cite{DBLP:journals/pvldb/ArocenaGCM15}, \textsc{Doctors}~\cite{DBLP:conf/icde/GeertsMPS14}, and \textsc{Deep} scenarios have been recently used to evaluate the performance of chase-based systems~\cite{DBLP:conf/pods/BenediktKMMPST17}. Their queries were especially tailored towards the evaluation of query answering tasks of such systems. 
Note that the \textsc{LUBM} benchmark~\cite{DBLP:journals/ws/GuoPH05} is a widely used standard benchmark for the evaluation of Semantic Web repositories. Its queries are designed to measure the performance of those repositories over large datasets. Strictly speaking, the \textsc{iBench} is a tool for generating schemas, constraints, and mappings for data integration tasks. However, in~\cite{DBLP:conf/pods/BenediktKMMPST17}, 40 queries were created for tests with the 
\textsc{iBench}. We therefore refer to these queries as \textsc{iBench}-CQs here. In summary, we have incorporated all queries that were either contained in the original benchmarks or created/adapted for the tests in~\cite{DBLP:conf/pods/BenediktKMMPST17}.

The goal of the Join Order Benchmark (JOB)~\cite{DBLP:journals/vldb/LeisRGMBKN18} was to evaluate the impact of a good join order on the performance of query evaluation in standard 
RDBMSs. Those queries were formulated over the real-world dataset Internet Movie Database (IMDB). All of the queries have between 3 and 16 joins. Clearly, as the goal was to measure the impact of a good join order, those 33 queries are of higher complexity, hence 7 out of the 33 queries have $\hw \geq 2$.

The Transaction Processing Performance Council (TPC) is a well-known non-profit organization that develops benchmarks for the evaluation of DBMSs.
Given their broad industry-wide relevance and since they reflect common workloads in decision support systems, we included the TPC-H~\cite{tpch} and the TPC-DS~\cite{tpcds} benchmarks.
In~\cite{DBLP:conf/pods/FischlGLP19}, we analyzed the TPC-H queries from the GitHub repository 
originally provided by Michael Benedikt and Efthymia Tsamoura~\cite{BenediktCQs} for the work on~\cite{DBLP:conf/pods/BenediktKMMPST17}.
Nevertheless, for this paper we downloaded the original dataset from~\cite{tpch} and extracted the queries according to the 
methodology introduced in Sections~\ref{sec:translation} -- \ref{sec:translation-hg}.
From the original set of 22 complex queries, we extracted 29 simple queries.
The TPC-DS benchamrk is more complex than TPC-H and it contains more queries.
Indeed, from the original set of 113 complex queries, we extracted 228 simple queries.

From SQLShare~\cite{DBLP:conf/sigmod/JainMHHL16}, a multi-year SQL-as-a-service experiment with 
a large set of real-world queries, we extracted 15,170 queries by considering all queries in the log files.
These queries are divided into two sets: materialized views and usual queries which could possibly make use of the materialized views.
Also, the whole dataset gathers data from different databases and the link between queries and databases is not explicitly defined.
In order to execute the experiments, we had to clean the queries from trivial errors impeding the parsing,
link the queries to the right database schema, incorporate the materialized views and resolve ambiguities in the query semantics.
After removing queries with complex syntactical errors, we obtained 12,483 queries.
As a next step we used the algorithm from Section~\ref{sec:extract-simple-queries} to obtain a set of CQs and,
after removing duplicates, we got a collection of 6,086 simple SQL queries.
From this set we eliminated 5,796 queries with $\leq$ 2 atoms (whose acyclicity is immediate)
and ended up with 290 queries.

The random queries were generated with a tool that stems from the work on query answering using views in~\cite{DBLP:journals/vldb/PottingerH01}. 
The query generator allows 3 options: chain/star/random queries. Since the former two types are trivially acyclic, 
we only used the third option. Here it is possible to supply several parameters for the size of the generated queries. In terms of the 
resulting hypergraphs, one can thus fix the number of vertices, number of edges and arity. We have generated 500 CQs with 5 -- 100 vertices, 
3 -- 50 edges and arities from 3 to 20.
These values correspond to the values observed for the {\it CQ Application\/} hypergraphs.
However, even though these size values have been chosen similarly, 
the structural properties of the hypergraphs in the two groups {\it CQ Application\/} and {\it CQ Random\/} 
differ significantly,
as will become  clear from our analysis in Section~\ref{sec:hg-stats}.

As was detailed in Section~\ref{sec:cspHypergraphs}, 
our benchmark currently contains \cspTot{} hypergraphs from CSP instances, 
out of which 1,953 instances were obtained 
from \url{xcsp.org} (see also \cite{xcsp}). 
These instances, 
in turn,
are divided into CSPs
from concrete applications, called \textit{CSP Application} in the sequel (\cspApp{} instances), 
and randomly generated CSPs, called \textit{CSP Random} below (\cspRand{} instances). 
In addition, we have \cspOth{} CSP instances from 
provided at \url{https://www.dbai.tuwien.ac.at/proj/hypertree/}, 
which we refer to as {\em CSP Other}.

Our HyperBench benchmark consists of the CQ and CSP instances converted to hypergraphs.
In Figure~\ref{fig:hg-sizes}, we show the number of vertices, the number of edges and the arity (i.e., the maximum size of the edges) as 
three important metrics of the size of each hypergraph. 
The smallest are those coming from {\it CQ Application\/} (most of them have up to 10 edges), while the 
hypergraphs coming from CSPs can be significantly larger (up to 2993 edges).
Although some hypergraphs are very big, more than 50\% of all hypergraphs have maximum arity less than 5. 
In Figure~\ref{fig:hg-sizes} we can easily compare the different types of hypergraphs, e.g.\ hypergraphs of arity greater than 20  only exist in the application classes;
the \emph{CSP Other} class contains the highest portion of 
hypergraphs with a big number of vertices and edges, etc.

\begin{figure}[htbp]
	\centering 
	\fbox{\includegraphics[width=0.45\textwidth]{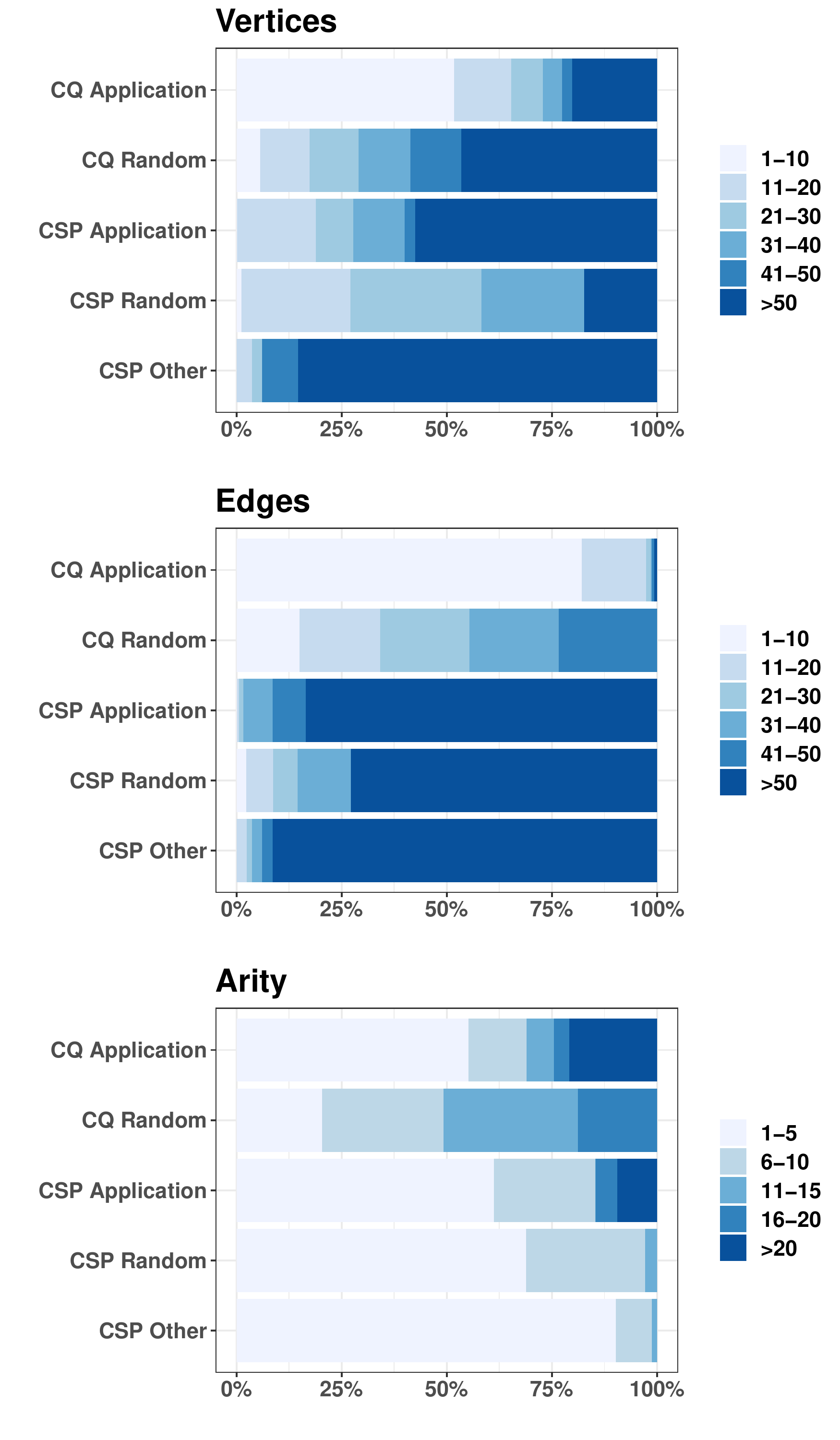}}
	\caption{Hypergraph Sizes} 
	\label{fig:hg-sizes}         
\end{figure}    

The hypergraphs and the results of our analysis can be accessed through our web tool, available at \url{http://hyperbench.dbai.tuwien.ac.at}.

\section{Experiments}
\label{sec:experiments}
In this section, we present the empirical results obtained with the HyperBench benchmark.
On the one hand, we want to get an overview of the hypertree width of 
the various types of hypergraphs in our benchmark (cf.\ Goal 2 in Section~\ref{sec:introduction}).
On the other hand, we want to find out how realistic 
the restriction to low values for certain hypergraph invariants is (cf.\ Goal 3 stated in Section~\ref{sec:introduction}). 
After this first analysis of the structural properties of the hypergraphs we collected,
we perform an evaluation of the different $\ghw$ algorithms presented in Section~\ref{sec:ghd-algs}.
Finally, we propose and evaluate two algorithms for computing approximated FHDs.

\subsection{Hypergraph Properties}
\label{sec:hg-stats}
In \cite{DBLP:conf/pods/FischlGP18,DBLP:journals/corr/abs-2002-05239}, several invariants of 
hypergraphs were used to make the problems
$\checkp{(\mathrm{GHD},k)}$  and $\checkp{(\mathrm{FHD},k)}$ 
tractable or, at least, easier to approximate.
We thus investigate the following properties (cf.\ Definitions \ref{def:bip}~--~\ref{def:vc}):
\begin{itemize}
	\item {\it Deg}: the degree of the underlying hypergraph
	\item {\it BIP}: the intersection size
	\item $c$-{\it BMIP}: the $c$-multi-intersection size for $c \in \{3,4\}$
	\item {\it VC-dim}: the VC-dimension 
\end{itemize}

\begin{table}[t]
	\centering 
	\caption{Properties of all benchmark instances}
	\label{tab:hg-props}%
	\begin{tabular}{crrrrr}
		\multicolumn{6}{c}{\it CQ Application} \\
		\toprule
		\multicolumn{1}{c}{$i$}  & \multicolumn{1}{c}{Deg} & 
		\multicolumn{1}{c}{BIP} & \multicolumn{1}{c}{3-BMIP} & 
		\multicolumn{1}{c}{4-BMIP} & \multicolumn{1}{c}{VC-dim} \\
		\midrule
		0	&	0	&	74	&	394	&	647	&	72	\\
		1	&	74	&	721	&	673	&	456	&	484	\\
		2	&	320	&	286	&	45	&	10	&	557	\\
		3	&	253	&	17	&	1	&	0	&	0	\\
		4	&	181	&	6	&	0	&	0	&	0	\\
		5	&	73	&	9	&	0	&	0	&	0	\\
		$>$5&	212	&	0	&	0	&	0	&	0	\\
		\bottomrule
	\end{tabular}%
	\quad
	\begin{tabular}{crrrrr}
		\multicolumn{6}{c}{\it CQ Random} \\
		\toprule
		\multicolumn{1}{c}{$i$}  & \multicolumn{1}{c}{Deg} & 
		\multicolumn{1}{c}{BIP} & \multicolumn{1}{c}{3-BMIP} & 
		\multicolumn{1}{c}{4-BMIP} & \multicolumn{1}{c}{VC-dim} \\
		\midrule
		0     & 0     & 1     & 16    & 49    & 0	\\
		1     & 1     & 17    & 77    & 125   & 20	\\
		2     & 15    & 53    & 90    & 120   & 133	\\
		3     & 38    & 62    & 103   & 74    & 240	\\
		4     & 31    & 63    & 62    & 42    & 106	\\
		5     & 33    & 71    & 47    & 28    & 1	\\
		$>$5  & 382   & 233   & 105   & 62    & 0	\\ 
		\bottomrule
	\end{tabular}
	\;
	\begin{tabular}{crrrrr}
		\multicolumn{6}{c}{\it CSP Application} \\
		\toprule
		\multicolumn{1}{c}{$i$}  & \multicolumn{1}{c}{Deg} & 
		\multicolumn{1}{c}{BIP} & \multicolumn{1}{c}{3-BMIP} & 
		\multicolumn{1}{c}{4-BMIP} & \multicolumn{1}{c}{VC-dim} \\
		\midrule
		0		&	0	&	0	&	596	&	597	&	0	\\
		1		&	0	&	1030&	459	&	486	&	0	\\
		2		&	596	&	59	&	34	&	7	&	1064\\
		3		&	1	&	0	&	1	&	0	&	26	\\
		4		&	1	&	0	&	0	&	0	&	0	\\
		5		&	2	&	0	&	0	&	0	&	0	\\
		$>$5	&	490	&	1	&	0	&	0	&	0	\\
		\bottomrule
	\end{tabular}
	\quad
	\begin{tabular}{crrrrr}
		\multicolumn{6}{c}{\it CSP Random} \\
		\toprule
		\multicolumn{1}{c}{$i$}  & \multicolumn{1}{c}{Deg} & 
		\multicolumn{1}{c}{BIP} & \multicolumn{1}{c}{3-BMIP} & 
		\multicolumn{1}{c}{4-BMIP} & \multicolumn{1}{c}{VC-dim} \\
		\midrule
		0     & 0     & 0     & 0     & 0     & 0	\\
		1     & 0     & 200   & 200   & 238   & 0	\\
		2     & 0     & 224   & 312   & 407   & 220	\\
		3     & 0     & 76    & 147   & 95    & 515	\\
		4     & 12    & 181   & 161   & 97    & 57	\\
		5     & 8     & 99    & 14    & 1     & 71	\\
		$>$5  & 843   & 83    & 29    & 25    & 0	\\
		\bottomrule
	\end{tabular}%
	\;
	\begin{tabular}{crrrrr}
		\multicolumn{6}{c}{\it CSP Other} \\
		\toprule
		\multicolumn{1}{c}{$i$}  & \multicolumn{1}{c}{Deg} & 
		\multicolumn{1}{c}{BIP} & \multicolumn{1}{c}{3-BMIP} & 
		\multicolumn{1}{c}{4-BMIP} & \multicolumn{1}{c}{VC-dim} \\
		\midrule
		0	&	0	&	0	&	1	&	6	&	0	\\
		1	&	0	&	7	&	36	&	39	&	0	\\
		2	&	1	&	36	&	23	&	16	&	50	\\
		3	&	5	&	29	&	20	&	21	&	25	\\
		4	&	19	&	10	&	2	&	0	&	0	\\
		5	&	4	&	0	&	0	&	0	&	0	\\
		$>$5&	53	&	0	&	0	&	0	&	0	\\ 
		\bottomrule
	\end{tabular}
\end{table}%

The results obtained from computing 
{\it Deg}, {\it BIP}, $3$-{\it BMIP},  $4$-{\it BMIP}, and {\it VC-dim}
for the hypergraphs in the HyperBench benchmark 
are shown in Table~\ref{tab:hg-props}.

Table~\ref{tab:hg-props} has to be read as follows: In the first column, we distinguish different values of the various hypergraph metrics. 
In the columns labeled ``Deg``, ``BIP``, etc., we indicate for how many instances each metric has a particular value. 
For instance, by the last row in the second column, 
only 212 non-random CQs
have degree~$>5$. Actually, for most CQs, the degree is less than 10.
Moreover, for the BMIP, already with intersections of 3 edges, we get 3-multi-intersection size $\leq 2$ 
for almost all non-random CQs. Also the VC-dimension is $\leq 2$. 

For CSPs, all properties may have higher values. However, we note a significant 
difference between randomly generated CSPs and the rest: 
For hypergraphs in the groups {\em CSP Application\/} and {\em CSP Other\/}, 
543 (46\%) hypergraphs have a high degree ($>$5), but nearly all instances have BIP or BMIP of less than 3. And most instances have a VC-dimension of at most 2. 
In contrast, nearly all random instances have a 
significantly higher degree (843 out of 863 instances with a degree $>5$). 
Nevertheless, many instances have small BIP and 
BMIP. For nearly all hypergraphs (838 out of 863) 
we have 4-multi-intersection size $\leq 4$. 
For 7 instances the computation of the VC-dimension
timed out. For all others, the VC-dimension is $\leq 5$ for random CSPs.
Clearly, as seen in Table~\ref{tab:hg-props}, the random CQs resemble the random CSPs a lot more than the CQ and CSP Application instances. For example, random CQs have similar to random CSPs high degree (382, corresponding 
to 76\%, with degree $>5$), higher BIP and BMIP. Nevertheless, similar to random CSPs, the values for BIP and BMIP are still small for many random CQ instances.

To conclude, for the proposed properties, in particular BIP/BMIP and VC-dimension,
most of the hypergraphs in our benchmark indeed have low values.

\medskip
\begin{figure}
	\centering 
	\fbox{\includegraphics[width=0.45\textwidth]{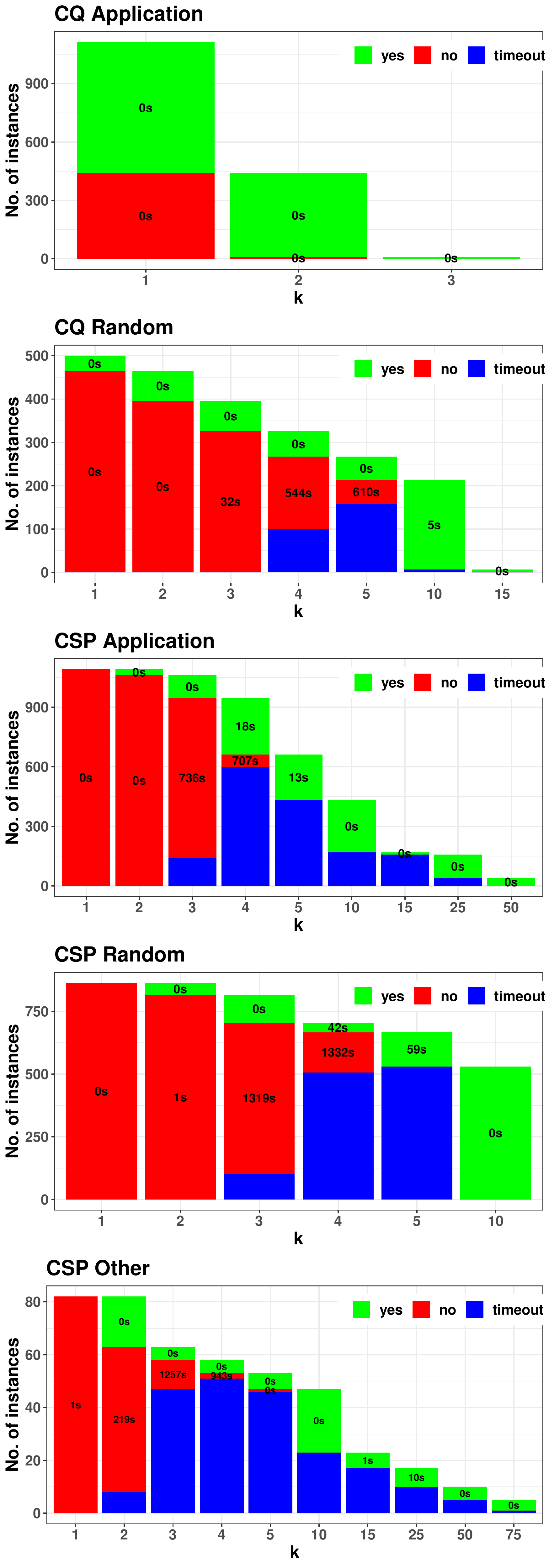}}
	\caption{HW analysis (labels are avg. runtimes in s)} 
	\label{fig:hw}      
\end{figure}

\subsection{Hypertree Width}
We have systematically applied the $\hw$-compu\-tation from~\cite{DBLP:journals/jea/GottlobS08} 
to all hypergraphs in the benchmark. The results are summarized in Figure~\ref{fig:hw}. 
In our experiments, we proceeded as follows. 
We used the same classification of instances we used in the previous experiments, i.e., we distinguish the following classes:
\emph{CQ Application}, \emph{CQ Random}, \emph{CSP Application}, \emph{CSP Random}, and \emph{CSP Other}.
For every hypergraph $H$, we first tried to solve the $\checkp{(\mathrm{HD},k)}$ problem for $k = 1$.
In case of {\em CQ Application\/}, we thus got \cqAppKOneYes{} yes-answers
and \cqAppKOneNo{} no-answers. The number in each bar indicates the average runtime to 
find these yes- and no-instances, respectively. 
Here, the average runtime was ``0'' (i.e., less than 1 second). 
For {\it CQ Random\/} we got 36 yes- and 464 no-instances with an average runtime below 1 second.
For all CSP-instances, we only got no-answers. 

In the second round, we tried to solve the $\checkp{(\mathrm{HD},k)}$ problem for 
$k = 2$ for all hypergraphs that yielded a no-answer for $k=1$. 
Now the picture is a bit more diverse:
\cqAppKTwoYes{} of the remaining \cqAppKOneNo{} CQs from {\it CQ Application\/} yielded a yes-answer in less than 1 second.
For the hypergraphs stemming from {\it CQ Random}, only 
68 instances yielded a yes-answer (in less than 1 second on average), while 396 
instances yielded a no-answer in less than 7 seconds on average.
The hypergraphs relative to CSP offer a different picture.
The classes \emph{CSP Application}, \emph{CSP Random} and \emph{CSP Other} have 29, 47 and 19 yes-instances, respectively.
Only 8 instances from \emph{CSP Other} give rise to a timeout (i.e., the program did not terminate within 3,600 seconds),
while all the other instances give a no-answer within the timeout.
Interestingly, the our $\hw$-algorithm gives a no-answer for 1877 instances of \emph{CSP Application} and \emph{CSP Random} in less than 1 second,
while it took the algorithm 219 seconds on average to answer ``no" for 55 instances of \emph{CSP Other}.
This shows that the class \emph{CSP Other} contains instances which are difficult to decompose.

This procedure is iterated by incrementing $k$ and running the $\hw$-computation for all instances, that either yielded a no-answer or a timeout in the previous round. For instance, for queries from {\it CQ Application\/}, 
one further round is needed after the second round.
In other words, we confirm the observation of low $\hw$, which was 
already made for CQs of arity $\leq 3$ in~\cite{DBLP:journals/pvldb/BonifatiMT17,DBLP:conf/www/BonifatiMT19,DBLP:conf/sigmod/PicalausaV11}.
For the hypergraphs stemming from {\em CQ Random} (resp.\ CSPs), 396 (resp.\ 1940) instances 
are left in the third round, of which 70 (resp.\ 232) yield a yes-answer in less than 1 second on average, 326 (resp.\ 1415) instances yield a no-answer in 32 (resp.\ 988) seconds on average and 
no (resp.\ 293) instances yield a timeout. 
Note that, as we increase $k$, 
the average runtime  and the percentage of timeouts first increase up to a certain point and then they decrease.
This is due to the fact that,
as we increase $k$, the number of combinations of edges to be considered in 
each $\lambda$-label (i.e., the function $\lambda_u$ at each node $u$ 
of the decomposition) increases. In principle, we have to test 
$\mathcal{O} (n^{k})$ combinations, where $n$ is the number of edges.
However, if $k$ increases beyond a certain point, then it gets easier to ``guess'' a $\lambda$-label since an increasing portion of the 
$\mathcal{O} (n^{k})$ possible combinations leads to a solution (i.e., 
an HD of desired width).

To answer the question in {\it Goal 2}, it is indeed the case that for a big number of instances,
the hypertree width is small enough to allow for efficient evaluation of CQs or CSPs: all instances of non-random CQs have $\hw \leq 3$ no matter 
whether their arity is bounded by 3 (as in case of SPARQL queries) or not; and a large portion (at least 1027, i.e., ca.\ 50\%) 
of all 2035 CSP instances have $\hw \leq 5$. 
In total, including random CQs, 2,427 (66.5\%) out of \hgTotal{} instances have $\hw \leq 5$. And, out of these, 
we could determine the exact hypertree width for 2,356 instances; the others may even have lower~$\hw$.

\subsection{Correlation Analysis}
Finally, we have analyzed the pairwise correlation between all properties. 
Of course, the different intersection sizes (BIP, 3-BMIP, 4-BMIP) are highly correlated.
Other than that, we observe quite a strong correlation of the arity with the number of vertices and the hypertree width.
Moreover, there is a significant correlation between number of vertices and arity and between number of vertices and hypertree width.
Clearly, the correlation between arity and hypertree width 
is mainly due to the CSP instances and the random CQs since, for non-random CQs, 
the $\hw$ never increases beyond~$3$,  independently of the arity.

A graphical presentation of all pairwise correlations is given in Figure~\ref{fig:corr}. Here, large, dark circles indicate a high correlation, 
while small, light circles stand for low correlation. 
Blue circles indicate a positive  correlation while red circles stand for a negative correlation.  
In~\cite{DBLP:conf/pods/FischlGP18}, it has been  argued that Deg, BIP, 3-BMIP, 4-BMIP and VC-dim are non-trivial restrictions to achieve tractability. It is interesting to note that, 
according to the correlations shown in Figure~\ref{fig:corr}, 
these properties have almost no impact on the hypertree width of our hypergraphs.
This underlines the usefulness of these restrictions in the sense that (a) they make the GHD computation and 
FHD approximation easier~\cite{DBLP:conf/pods/FischlGP18} but (b) low values of degree, (multi-)intersection-size, or VC-dimension do not pre-determine low values of the widths.

\begin{figure}[htbp]
	\centering 
	\fbox{\includegraphics[width=0.5\textwidth]{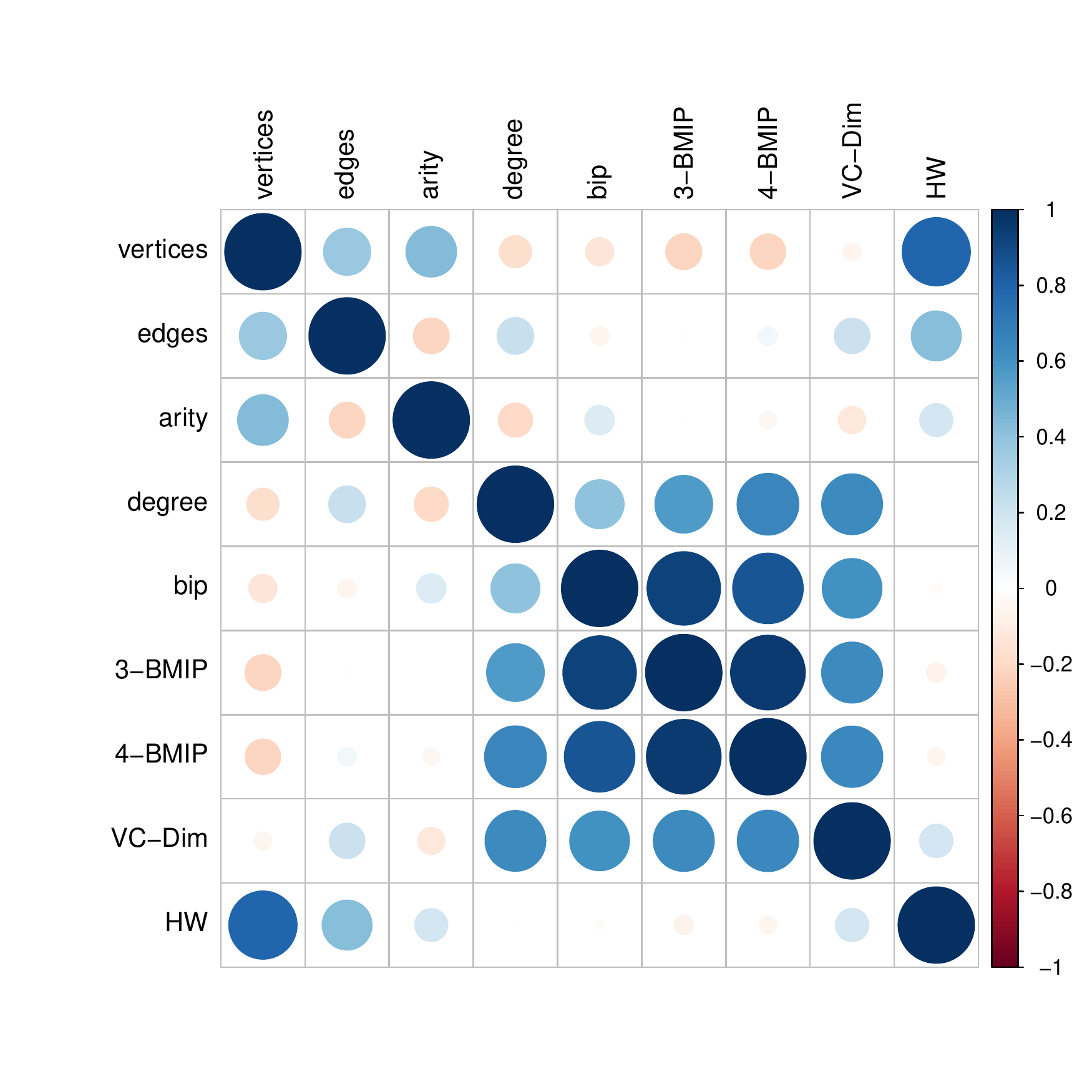}}
	\caption{Correlation analysis} 
	\label{fig:corr}     
\end{figure}

\subsection{Comparison of $\ghw$ Algorithms}
Here we report on empirical results for the three $\ghw$-algorithms described in Section~\ref{sec:ghd-algs}.
We have run the programs on each hypergraph from the HyperBench up to hypertree width $6$, 
trying to get a smaller $\ghw$ than $\hw$. 
We have thus run the $\ghw$-algorithms with the following parameters:
for all hypergraphs $H$ with $\hw(H) = k$ (or $\hw \leq k$ and, due to timeouts, we do not know if $\hw \leq k-1$ holds), 
where $k \in \{3,4,5,6\}$, try to solve the $\checkp{(\mathrm{GHD},k-1)}$ problem. 
In other words, we just tried to improve the width by 1. 
Clearly, for $\hw(H) \in \{1,2\}$, no improvement is possible since, in this case, $\hw(H) = \ghw(H)$ holds.

In Table~\ref{tab:ghw_instances_solved}, for each algorithm, we report on the number of ``successful'' attempts
to solve the $\checkp{(\mathrm{GHD},k-1)}$ problem for hypergraphs with $\hw =k$.
Here ``successful'' means that the program terminated within 1 hour. 
For instance, for the 310 hypergraphs with $\hw = 3$ in the HyperBench, 
\globalbip{} terminated in 128 cases (i.e., 41\%) when trying to solve $\checkp{(\mathrm{GHD},2)}$.
The average runtime of these ``successful'' runs was 537 seconds.
For the 386 hypergraphs with $\hw = 4$, \globalbip{} terminated in 137 cases (i.e., 35\%) 
with average runtime 2809 when trying to solve the $\checkp{(\mathrm{GHD},3)}$ problem.
For the 886 hypergraphs with $\hw \in \{5,6\}$, \globalbip{} only terminated in 13 cases (i.e., 1.4\%). 
Overall, it turns out that the set $f(H,k)$ may be very big (even though it is 
polynomial if $k$ and $i$ are constants). Hence, $H'$ can become  considerably bigger than $H$. 
This explains the frequent timeouts in the \globalbip{} column in Table~\ref{tab:ghw_instances_solved}.

\begin{table}
	\centering 
	\caption{GHW algorithms with avg. runtimes in s}
	\label{tab:ghw_instances_solved}    
	\begin{tabular}{cccccccc} 
		\toprule
		\multirow{2}{*}{$\hw \rightarrow \ghw$} & \multirow{2}{*}{\it Total}
		& \multicolumn{2}{c}{\globalbip{}} & \multicolumn{2}{c}{\localbip{}} & \multicolumn{2}{c}{\balsep{}} \\
		
		& & \emph{yes (s)} & \emph{no (s)} & \emph{yes (s)} & \emph{no (s)} & \emph{yes (s)} & \emph{no (s)} \\
		\midrule
		$3 \rightarrow 2$	& 310 & -		 & 128 (537)  & -		& 195 (162)	& - & 307 (12)	\\
		$4 \rightarrow 3$	& 386 & -		 & 137 (2809) & - 		& 54 (2606) & - & 249 (54)	\\
		$5 \rightarrow 4$	& 427 & - 		 & -		  & - 		& - 		& - & 148 (13)	\\ 
		$6 \rightarrow 5$ 	& 459 & 13 (162) & -		  & 13 (60) & - 		& - & 180 (288)	\\ 
		\bottomrule
	\end{tabular}
\end{table}

The results obtained with \localbip{} are shown in the corresponding column.
Interestingly, for the hypergraphs with $\hw = 3$, the  ``local'' computation performs significantly better
(namely 63\% solved with average runtime 162 seconds rather than 41\% with average runtime 537 seconds). 
In contrast, for the hypergraphs with $\hw = 4$, the ``global'' computation is significantly more successful.
For $\hw \in \{5,6\}$, the ``global'' and ``local'' computations are equally bad.
A possible explanation for the reverse behavior of ``global'' and ``local'' computation in case of 
$\hw = 3$ as opposed to $\hw = 4$ is that the restriction of the ``global'' set 
$f(H,k)$ of subedges to the ``local'' set $f_u(H,k)$ at each node $u$ seems to be
quite effective for the hypergraphs with $\hw = 3$.
In contrast, the additional cost of having to compute $f_u(H,k)$ at each node $u$ becomes counter-productive,
when the set of subedges thus eliminated is not significant.
It is interesting to note that the sets of solved instances 
of the global computation and the local computation are
incomparable, i.e., in some cases one method is better, while
in other cases the other method is better.

If we look at the number of solved instances in Table~\ref{tab:ghw_instances_solved}, we see that the 
recursive algorithm via balanced separators (reported in the last column labeled \balsep{}) has the least number 
of timeouts due to the fast identification of negative instances (i.e., those with no-answer), 
where it often detects quite fast that a given hypergraph 
does not have a balanced separator of desired~width.
As $k$ increases, the performance of the balanced separators approach deteriorates. 
This is due to $k$ in the exponent of the running time of our algorithm, i.e.\ we need to check 
for each of the possible $\mathcal{O}(n^{k+1})$ 
combinations of $\leq k$ edges if it constitutes a balanced separator. 
Note that the balanced separators approach only terminated in case of no-answers.

\begin{table}
	\centering 
	\caption{GHW of instances with average runtime in s}
	\label{tab:ghw}
	\begin{tabular}{cccc}
		\toprule
		$\hw \rightarrow \ghw$ & yes & no & timeout \\
		\midrule
		$3 \rightarrow 2$	& 0			& 309 (10)	& 1		\\
		$4 \rightarrow 3$	& 0			& 262 (57)	& 124	\\
		$5 \rightarrow 4$	& 0			& 148 (13)	& 279	\\
		$6 \rightarrow 5$	& 18 (129)	& 180 (288)	& 261	\\
		\bottomrule
	\end{tabular}
\end{table}

We now look at Table \ref{tab:ghw}, where we report for all hypergraphs with $\hw \leq k$ and $k \in \{3,4,5,6\}$,
whether $\ghw \leq k-1$ could be verified. 
To this end, we run our three algorithms (\globalbip{}, \localbip{} and \balsep{}) in parallel and stop the computation,
as soon as one terminates (with answer ``yes'' or ``no'').
The number in parentheses refers to the average runtime 
needed by the fastest of the three algorithms in each case.
A timeout occurs if none of the three algorithms terminates within 3,600 seconds.
It is interesting to note that in the vast majority of cases, no improvement of the width is possible when we switch from $\hw$ to $\ghw$:
in 97\% of the solved cases with $\hw \leq 6$, which form 65\% of all instances, $\hw$ and $\ghw$ have identical values.
Actually, we think that the high percentage of the {\em solved cases\/} 
gives a more realistic picture than the percentage of {\em all cases\/}
for the following reason: our algorithms (in particular, the ``global'' and ``local'' computations) need particularly long time for negative instances. This is due to the fact that in a negative case, ``all'' possible choices of $\lambda$-labels for a node $u$ in the GHD have to be tested before we can be sure that no GHD of $H$ (or, equivalently, no HD of $H'$) of desired width exists. 
Hence, it seems plausible that the timeouts are mainly due to negative instances.
This also explains why our new \balsep{} algorithm, which is 
particularly well suited for negative instances,
has the least number of timeouts.

We conclude this section with a final observation: in Figure~\ref{fig:hw}, we had many cases,
for which only some upper bound $k$ on the $\hw$ could be determined, namely those cases, where the attempt to solve $\checkp{(\mathrm{HD},k)}$
yields a yes-answer and the attempt to solve $\checkp{(\mathrm{HD},k-1)}$ gives a timeout.
In several such cases, we could get (with the balanced separator approach) a 
no-answer for the $\checkp{(\mathrm{GHD},k-1)}$ problem, which implicitly gives a 
no-answer for the problem $\checkp{(\mathrm{HD},k-1)}$. 
In this way, our new $\ghw$-algorithm is also profitable for the $\hw$-computation:
for 827 instances with $\hw \leq 6$, we were not able to determine the exact hypertree width. 
Using our new $\ghw$-algorithm, we closed this gap for 297 instances; for these instances $\hw = \ghw$ holds. 

To sum up, we now have a total of 2,356 (64.5\%) instances for which we determined the exact $\hw$
and a total of 1,984 instances (54.4\%) for which we determined the exact $\ghw$.
Out of these, 1,968 instances had identical values for $\hw$ and $\ghw$.
In 16 cases, we found an improvement of the width by 1 when moving from $\hw$ to $\ghw$, namely from $\hw = 6$ 
to $\ghw = 5$. In 2 further cases, we could show $\hw \leq 6$ and $\ghw \leq 5$, but the attempt to check 
$\hw = 5$ or $\ghw = 4$ led to a timeout.
Hence, in response to {\it Goal 6}, $\hw$ is equal to $\ghw$ in 54.4\% of the cases if we consider all instances
and in 68.2\% of the cases (1,968 of 2,886) with small width ($\hw \leq 6$). However, if we consider the fully 
solved cases (i.e., where we have the precise value of $\hw$ and $\ghw$), then $\hw$ and $\ghw$ coincide in 
99.2\% of the cases (1,968 of 1,984).

\subsection{Fractionally Improved Decompositions}
\label{sec:fhd-algs}
The algorithms for computing FHDs in the literature are very expensive and even the tractability result presented in~\cite{DBLP:conf/pods/FischlGP18} 
involves a double exponential ``constant''.
Here we propose two algorithms for computing an FHD of a hypergraph when we already have a GHD: \improvehd{}, and \fracimprovehd{}.
They differ in the compromise between computational cost and quality of the approximation.

The first algorithm we present is based on a simple observation: given an (G)HD, we could substitute its integral edge covers with fractional edge covers and obtain an FHD.
Formally, let $\mathcal{D} = \langle T, (B_u)_{u \in T}, (\lambda_u)_{u \in T} \rangle$ be either a GHD or an HD.
Our algorithm \improvehd{} computes an FHD $\mathcal{D'} = \langle T', (B'_u)_{u \in T'}, (\gamma_u)_{u \in T'} \rangle$ where:
\begin{itemize}
	\item $T' = T$.
	\item $\forall u \in T', B'_u = B_u$.
	\item $\forall u \in T', \gamma_u$ is a fractional edge cover of $B'_u$.
\end{itemize}
To obtain such $\mathcal{D'}$, it is sufficient to iterate over the nodes of $\mathcal{D}$ and substitute $\lambda_u$ with $\gamma_u$.
Since computing a fractional edge cover is polynomial and we assume to have already computed an HD to start with, the whole algorithm is very efficient.
Nevertheless, it is clear that there is a strong dependence on the starting HD.
This is unsatisfactory and so we devised a more sophisticated algorithm.

The algorithm we describe here gets rid of the dependence on a particular HD and computes a fractionally improved (G)HD with a fixed improvement threshold.
given a hypergraph $H$ and two numbers $k,k' \geq 1$, where $k$ is an upper bound on the $\hw$ and $k'$ the desired fractionally improved $\hw$.
We search for an FHD $\mathcal{D}'$ with  $\mathcal{D}' = \mathit{SimpleImproveHD}(\mathcal{D})$ for some HD $\mathcal{D}$ of $H$ with $\mathit{width}(\mathcal{D})  \leq k$ and $\mathit{width}(\mathcal{D}') \leq k'$.
In other words, this algorithm searches for the best fractionally improved HD over all HDs of width $\leq k$. 
Hence, the result is independent of any concrete HD.

The experimental results with these algorithms for computing fractionally improved HDs 
are summarized in Tables~\ref{tab:fhw-improve} and~\ref{tab:fhw-fracimprove}. 
We have applied these algorithms to all 
hypergraphs for which $\hw \leq k$ 
with $k \in \{2,3,4,5,6\}$ is known from Figure~\ref{fig:hw}.
The various columns of the Tables~\ref{tab:fhw-improve} and~\ref{tab:fhw-fracimprove} are as follows: the first column (labeled $\hw$) refers to the 
(upper bound on the) $\hw$ according to Figure~\ref{fig:hw}. The next 3 columns, labeled $\geq 1$, 
$[0.5,1)$, and $[0.1,0.5)$ tell us, by how much the width can be improved (if at all) 
if we compute an FHD by one of the two algorithms. 
We thus distinguish the 3 cases if, for a hypergraph of $\hw \leq k$, 
we manage to construct an FHD of width $k-c$ for 
$c \geq 1$, $c \in [0.5,1)$, or  $c \in [0.1,0.5)$. The column with label ``no'' refers to the cases where no
improvement at all or at least no improvement by $c \geq 0.1$ was possible.
The last column counts the number of timeouts. 

For instance, in the first row of Table~\ref{tab:fhw-improve}, we see that
(with the \improvehd{} algorithm and starting from the HD obtained by the $\hw$-computation of Figure~\ref{fig:hw})
out of 595 hypergraphs 
with $\hw = 2$, no improvement was possible in 419 cases. In the remaining 176 cases, 
an improvement to a width of at most $2 - 0.5$ was possible in 40 cases and an improvement to $k-c$ with 
$c \in [0.1,0.5)$ was possible in 136 cases. For the hypergraphs with $\hw = 3$ in Figure~\ref{fig:hw}, 
almost half of the hypergraphs (141 out of 310) allowed at least some improvement, in particular, 104 by 
$c \in [0.5,1)$ and 12 even by at least 1. The improvements achieved for the hypergraphs with $\hw \leq 4$
and $\hw \leq 5$ are less significant.

\begin{table}[t]
	\centering 
	\caption{Instances solved with \improvehd{}}
	\label{tab:fhw-improve}%
	\begin{tabular}{crrrrr}
		\toprule
		\multicolumn{1}{c}{$\hw$}  & \multicolumn{1}{c}{$\geq1$} & 
		\multicolumn{1}{c}{$[0.5,1)$} & \multicolumn{1}{c}{$[0.1,0.5)$} & \multicolumn{1}{c}{no} & \multicolumn{1}{c}{timeout} \\
		\midrule
		2   & 0     & 136   & 40    & 419   & 0 \\
		3	& 12 	& 104	& 25 	& 169 	& 0 \\
		4	& 9 	& 55 	& 11	& 311 	& 0 \\
		5	& 20	& 14	& 11	& 382 	& 0 \\
		6   & 12    & 60    & 80    & 307   & 0 \\
		\bottomrule 
	\end{tabular}%
\end{table}%

\begin{table}[t]
	\centering 
	\caption{Instances solved with \fracimprovehd{}}
	\label{tab:fhw-fracimprove}%
	\begin{tabular}{crrrrr}
		\toprule
		\multicolumn{1}{c}{$\hw$}  & \multicolumn{1}{c}{$\geq1$} & 
		\multicolumn{1}{c}{$[0.5,1)$} & \multicolumn{1}{c}{$[0.1,0.5)$} & \multicolumn{1}{c}{no} & \multicolumn{1}{c}{timeout} \\
		\midrule
		2   & 0   & 194 & 46 & 353  & 2   \\
		3   & 14  & 116	& 21 & 135  & 24  \\
		4   & 11  & 81	& 2  & 8  	& 284 \\
		5   & 18  & 126	& 59 & 2  	& 222 \\
		6   & 28  & 149	& 95 & 4  	& 183 \\
		\bottomrule
	\end{tabular}%
\end{table}%

The results obtained with our \fracimprovehd{} implementation are displayed in  Table~\ref{tab:fhw-fracimprove}. 
We see that the number of hypergraphs which allow for a fractional improvement of the width by at least 0.5 
or even by 1 is often bigger than with \improvehd{} -- in particular in the cases where $k' \leq k$ with 
$k \in \{4,5\}$ holds. 
In the other cases, the results obtained with the naive \improvehd{} algorithm 
are not much worse than with the more sophisticated \fracimprovehd{} algorithm.

\section{Conclusion}
\label{sec:conclusion}

In this work, we have presented HyperBench, a new and comprehensive benchmark of hypergraphs derived from CQs and CSPs from various areas, together with the results of extensive empirical analyses with this benchmark.

\medskip

\noindent
{\bf Lessons learned.} The empirical study has brought many insights.
Below, we summarize the most important lessons learned from our studies.

\smallskip
$\bullet$ \
The finding of \cite{DBLP:journals/pvldb/BonifatiMT17,DBLP:conf/www/BonifatiMT19,DBLP:conf/sigmod/PicalausaV11} that 
non-random CQs have low hypertree width 
has been confirmed by our analysis, even if
(in contrast to SPARQL queries)
the arity of the CQs  is not bounded  by 3. 
For random CQs and CSPs, we have detected a correlation between the arity and the hypertree width, although 
also in this case, the increase of the $\hw$ with increased arity is not dramatic. 

\smallskip
$\bullet$ \
In \cite{DBLP:conf/pods/FischlGP18}, several hypergraph invariants were identified, which make the computation of GHDs and the approximation of FHDs tractable. We have seen that, at least for 
non-random instances, these invariants indeed tend to have low values.

\smallskip
$\bullet$ \
The reduction of the $\ghw$-computation problem to the $\hw$-computation problem in case of low intersection size
turned out to be more problematical than the theoretical trac\-tability results from \cite{DBLP:conf/pods/FischlGP18} had suggested. Even the improvement by ``local'' computation of the additional subedges did not help much.  
However, we were able to improve this significantly by presenting 
a new algorithm based on ``balanced separators''. In particular for negative instances (i.e., those with a no-answer), this approach proved very effective.

\smallskip
$\bullet$ \
An additional benefit of the new $\ghw$-algorithm based on ``balanced separators'' is that it allowed us to also fill gaps in the $\hw$-computation. Indeed, in several cases, 
we managed to verify 
$\hw \leq k$ for some $k$ but we could not show $\hw \not\leq k-1$, due to a timeout for
$\checkp{(\mathrm{HD},k-1)}$.
By establishing $\ghw \not\leq k-1$ with our 
new GHD-algorithm, we have implicitly shown $\hw \not\leq k-1$. This allowed us to 
compute the exact $\hw$ of many further~hypergraphs.

\smallskip
$\bullet$ \
Most surprisingly, the discrepancy between $\hw$ and $\ghw$ is much lower than expected. 
Theoretically, only the upper bound $\hw \leq 3 \cdot \ghw + 1$ is known. 
However, in practice, when considering hypergraphs of $\hw \leq 6$,
we could show that in circa 54\% of all cases, $\hw$ and $\ghw$ are simply identical. 
Moreover, in {\em all\/} cases when one of our implementations of $\ghw$-computation terminated on instances
with $\hw \leq 5$, we got identical values for $\hw$ and $\ghw$. 

\medskip

\noindent
{\bf Future work.} Our empirical  study has also given us many hints for future directions of research. We find the following tasks particularly urgent and/or rewarding. 

\smallskip
$\bullet$ \
So far, we have only implemented the $\ghw$-computation in case of low
intersection size. In~\cite{DBLP:conf/pods/FischlGP18}, tractability of the  
$\checkp{(\mathrm{GHD},k)}$ problem was also proved for the more relaxed bounded multi-intersection size. Our empirical results
in Table~\ref{tab:hg-props}
show that, apart from the random CQs and random CSPs, 
the 3-multi-intersection size is $\leq 2$ in almost all cases. It seems therefore worthwhile to 
implement and test also the BMIP-algorithm from~\cite{DBLP:conf/pods/FischlGP18}. 

\smallskip
$\bullet$ \ 
The three approaches for $\ghw$-computation presented here turned out to have complementary strengths and weaknesses. 
This was profitable when running all three algorithms in parallel and taking the result of the first one that terminates (see Table~\ref{tab:ghw}). 
In the future, we also want to implement a more sophisticated combination of the various approaches: for instance, one could try to apply our new ``balanced separator'' algorithm recursively only down to a certain recursion depth (say depth 2 or 3) to split a big given hypergraph into smaller subhypergraphs and then continue with the 
``global'' or ``local'' computation from 
Section~\ref{sec:ghd-algs}. First promising results in this direction have recently been obtained in~\cite{DBLP:conf/ijcai/GottlobOP20}.

\smallskip
$\bullet$ \
Our new approach to $\ghw$-computation via ``balanced separators'' proved quite effective in our experiments. However, further theoretical underpinning 
of this 
approach is missing. The empirical results obtained for our 
new GHD algorithm via balanced separators suggest that the
number of balanced separators is often drastically 
smaller than the number of arbitrary separators. 
We want to determine a realistic upper bound 
on the number of balanced separators 
in terms of $n$ (the number of edges) and $k$ (an upper bound on the width). This will then allow us to compute also a realistic upper bound on the runtime of this new algorithm.

\smallskip
$\bullet$ \
We want to further extend the HyperBench 
benchmark and tool in several directions. 
We will thus incorporate further implementations of decomposition algorithms from the literature 
such as the GHD- and FHD computation in \cite{DBLP:journals/ipl/MollTT12} or the polynomial-time FHD computation for hypergraphs of bounded multi-intersection size in~\cite{DBLP:journals/corr/abs-2002-05239}. 
Moreover, we will continue to fill in hypergraphs from further sources of CSPs and CQs. For instance, in 
\cite{DBLP:journals/tods/AbergerLTNOR17,DBLP:conf/pods/CarmeliKK17,DBLP:conf/icde/GhionnaGGS07,DBLP:conf/cikm/GhionnaGS11}  
a collection of CQs for the experimental evaluations in those papers is mentioned. 
We will invite the authors to disclose these CQs and incorporate them into the HyperBench benchmark. 

\smallskip
$\bullet$ \
Finally, we want to make use of HyperBench to test the practical feasibility of using decompositions to evaluate CQs and solving CSPs.
To this end, we will extend our collection of hypergraphs with the data of these problems.
In other words, we want to include in our benchmark the relations corresponding to the database for CQs and constraints for CSPs.
Such a study would have as a primary goal to assess the usefulness of decompositions in solving related problems,
but it would also help identify which characteristics should decompositions have to serve this purpose.

\begin{acks}
This work was supported by the Austrian Science Fund (FWF):P30930-N35 in the context
of the project ``HyperTrac''. 
Georg Gottlob is a Royal Society Research Professor and acknowledges support by the Royal Society 
for the present work in the context of the project "RAISON DATA" 
(Project reference: RP\textbackslash R1\textbackslash 201074).
Davide Mario Longo's work was also supported by the FWF project W1255-N23.

We would like to thank Angela Bonifati, Wim Martens, and Thomas Timm for sharing most of the hypergraphs with $\hw \geq 2$ from their work~\cite{DBLP:journals/pvldb/BonifatiMT17,DBLP:conf/www/BonifatiMT19} and for their effort in anonymizing these hypergraphs, which was required by the license restrictions. 
\end{acks}

\bibliographystyle{ACM-Reference-Format}
\bibliography{hyperbench}

\end{document}